\documentclass[11pt,reqno]{amsart}
\allowdisplaybreaks[4]
\usepackage{graphicx} %We can use any other package if it is necessary
\usepackage{epsfig}
\usepackage{amssymb}
\usepackage{amsmath}
\usepackage{cite}
\usepackage{braket}
\makeatletter

\newcommand{\Rmnum}[1]{\expandafter\@slowromancap\romannumeral #1@}
\makeatother
\newtheorem{theorem}{Theorem}

\newtheorem{proposition}[theorem]{Proposition}
\newtheorem{corollary}[theorem]{Corollary}

\newtheorem{lemma}[theorem]{Lemma}
\newtheorem{remark}[theorem]{Remark}
\newtheorem{example}{Example}
\newcommand{\be}{\begin{equation}}
\newcommand{\ee}{\end{equation}}
\newcommand{\al}{\alpha}

\newcommand{\La}{\Lambda}
\newcommand{\la}{\lambda}

\newcommand{\bft}{\mathbf{t}}
\newcommand{\bfx}{\mathbf{x}}
\newcommand{\bfy}{\mathbf{y}}
\newcommand{\e}{\mathrm{e}}
\newcommand{\Res}{\mathrm{Res}}

\newcommand{\Z}{\mathbb{Z}}
\newcommand{\C}{\mathbb{C}}
\newcommand{\F}{\mathcal{F}}

\newcommand{\p}{\partial}

\newcommand{\LL}{\mathcal{L}}
\newcommand{\fL}{\mathfrak{L}}

\begin{document}

\title{One reduction of the modified Toda hierarchy}
\author{Jinbiao Wang$^1$, Wenchuang Guan$^1$, Mengyao Chen$^1$, Jipeng Cheng$^{1,2 *}$}
\dedicatory {$^{1}$School of Mathematics, China University of Mining and Technology, Xuzhou, Jiangsu 221116,  China\\
$^{2}$ Jiangsu Center for Applied Mathematics (CUMT), \ Xuzhou, Jiangsu 221116, China}
\thanks{*Corresponding author. Email: chengjp@cumt.edu.cn,~ chengjipeng1983@163.com.}

\begin{abstract}
    The modified Toda (mToda) hierarchy is a two-component generalization of the 1-st modified KP (mKP) hierarchy, which connects the Toda hierarchy via Miura links and has two tau functions. Based on the fact that the mToda and 1-st mKP hierarchies share the same fermionic form, we firstly construct the reduction of the mToda hierarchy $L_1(n)^M=L_2(n)^N+\sum_{l\in\Z}\sum_{i=1}^{m}q_{i,n}\La^lr_{i,n+1}\Delta$ and $(L_1(n)^M+L_2(n)^N)(1)=0$, called the generalized bigraded modified Toda hierarchy, which can be viewed as a new two-component generalization of the constrained mKP hierarchy $\fL^k=(\fL^k)_{\geq 1}+\sum_{i=1}^m \mathfrak{q}_i\p^{-1}\mathfrak{r}_i\p$. Next the relation with the Toda reduction $\LL_1(n)^M=\LL_2(n)^{N}+\sum_{l\in \mathbb Z}\sum_{i=1}^{m}\tilde{q}_{i,n}\Lambda^l\tilde{r}_{i,n}$ is discussed. Finally we give equivalent formulations of the Toda and mToda reductions in terms of tau functions.\\
    \textbf{Keywords}: the modified Toda hierarchy; the constrained modified KP hierarchy; tau functions; Lax equations.\\
    {\bf MSC 2020}: 35Q53, 37K10, 35Q51 \\
    {\bf PACS}: 02.30.Ik
\end{abstract}

\maketitle

\section{Introduction}

The KP theory \cite{DJKM1983,Mulase1994,vanMoerbeke1994,JimboMiwa1983,JM2000} plays an important role in mathematical and theoretical physics, such as differential equations, Lie algebras, enumerative geometry and quantum gravity (see, e.g., \cite{JimboMiwa1983,JM2000,ABDKS20252,AA2015',Harnad2021,Kac2003,Kac2023}). In particular, multicomponent generalizations and reductions are crucial in the applications of the KP theory \cite{Kac1998,Kac2003,Kac2023,TakebeZabrodin2025}. Here we would like to discuss the reduction of the modified Toda (mToda) hierarchy\cite{RuiCheng2024}, which is the two-component generalization of the 1-st mKP hierarchy and connects with the Toda hierarchy \cite{Takasaki2018} (i.e., the two-component generalization of the KP hierarchy) by Miura links\cite{RuiCheng2024}.

The mToda hierarchy \cite{RuiCheng2024} is given by the Lax operators below
\[
    L_1(n,\bft,\La)=u(n,\bft)\La+\sum_{j\geq0}u_j(n,\bft)\La^{-j},\qquad L_2(n,\bft,\La)=\bar{u}(n,\bft)\La^{-1}+\sum_{j\geq0}\bar{u}_j(n,\bft)\La^j,
\]
satisfying the Lax equations 
\begin{equation*}
    \p_{t_k^{(1)}}L_a=[(L_1^k)_{\Delta,\geq1},L_a],\qquad \p_{t_k^{(2)}}L_a=[(L_2^k)_{\Delta^*,\geq1},L_a],\quad a=1,2,
\end{equation*} 
where $\bft=(\bft^{(1)},\bft^{(2)})$ and $\bft^{(a)}=(t^{(a)}_1,t^{(a)}_2,\cdots)$. Here $\La$ is the shift operator defined by $\La(f(n))=f(n+1)$, $\Delta=\La-1$ and $\Delta^*=\La^{-1}-1$ (see Section \ref{Sect:BieqtoLax} for more details). The mToda hierarchy contains the mToda lattice equation \cite{Hirota2004,Geng2003,RuiCheng2024}
\[
    \p_x\p_y\varphi(n)+\big(\e^{\varphi(n+1)-\varphi(n)}-\e^{\varphi(n)-\varphi(n-1)}\big)\p_x\varphi(n)=0,
\]
which is the reason why it is called the mToda hierarchy. Compared with the Toda hierarchy, there exist two tau functions $\tau_0(n,\bft)$ and $\tau_1(n,\bft)$ for the mToda hierarchy \cite{RuiCheng2024}. In the mToda hierarchy, it is more convenient to discuss the reductions involving adjoint operation $*$ defined by $(\sum_i a_i\La^i)^*=\sum_i\La^{-i}a_i$. For example, the B-Toda and C-Toda hierarchies are defined as the sub-hierarchies of the mToda hierarchy by requiring the constraints on the mToda Lax operators $L_1$ and $L_2$ such that
\begin{itemize}
    \item B-Toda\cite{Krichever2023,Guan2025}: $L_1^*\cdot (\La-\La^{-1})=(\La-\La^{-1})\cdot L_2$;
    \item C-Toda\cite{Krichever2022}: $L_1^*=L_2$.
\end{itemize}
Another important reduction of the mToda hierarchy is the bigraded mToda hierarchy \cite{Yang2024,Kac2023},
\begin{equation}
    L_1^M=L_2^N,\qquad (L_1^M+L_2^N)(1)=0,\qquad M,N\geq1,\label{eq:constmBTH}
\end{equation}
which corresponds to the loop algebra $sl_{M+N}[\la,\la^{-1}]$. In this paper, In this paper, we further generalize the bigraded mToda hierarchy by proposing a new reduction.

Note that the  mToda and mKP hierarchies share the same fermionic form \cite{Kac2018,RuiCheng2024,Yang2024}
\begin{equation}\label{eq:fermimKP}
    S_0(\tau_0\otimes\tau_1)=\tau_1\otimes\tau_0, \quad \tau_0\in\F_l,~\tau_1\in\F_{l+1},
\end{equation}
where $S_k=\sum_{j\in\Z+1/2}\psi^+_{j}\otimes\psi^-_{-j+k}$ with the charged free fermions $\psi^\pm_{j}~(j\in\Z+1/2)$ satisfying $\psi^\la_i\psi^\mu_j+\psi^\mu_j\psi^\la_i=\delta_{\la,-\mu}\delta_{i+j,0}$ for $\la,\mu=\pm$, and $\F_l$ is the subspace of Fock space $\F$ with the charge $l$ (please see Section \ref{Sect:BieqtoLax} for more details). On the other hand, the  $(k,m)$-constrained mKP (cmKP) hierarchy \cite{Cheng2018,WuCheng2022,Oevel1998} is given by the following constraint on the mKP Lax operator $\fL=\p+\sum_{j\geq0}v_j\p^{-j}$,
\begin{equation}\label{eq:constKP}
    \fL^k=(\fL^k)_{\geq 1}+\sum_{j=1}^m \mathfrak{q}_j\p^{-1}\mathfrak{r}_j\p,
\end{equation}
where  $v_j, \mathfrak{q}_j, \mathfrak{r}_j$ are functions of $\bfx=(x_1,x_2,\cdots)$, $\p=\p_x$ and $(\sum_ia_i\p^i)_{\geq1}=\sum_{i\geq1}a_i\p^i$. Here the cmKP hierarchy can be viewed as the generalization of the loop algebra $sl_k[\la,\la^{-1}]$-reduction of the mKP hierarchy \cite{Yang2024}, i.e., $\fL^k=(\fL^k)_{\geq 1}$. In this paper the mKP hierarchy always means the 1-st mKP hierarchy \cite{Cheng2018,Oevel1993}. The  cmKP hierarchy contains many famous integrable differential equations \cite{Kundu1995}, such as the  derivative nonlinear Schr{\"o}dinger equation, Gerdjikov-Ivanov equation and Chen-Lee-Liu equation. It will be proved in Section \ref{Sect:BieqtoLax} that the fermionic form of the cmKP hierarchy \eqref{eq:constKP} is given by
\[
    S_k(\tau_0\otimes\tau_1)=\sum_{1\leq i\leq m}\rho_i\otimes\sigma_i,\qquad\tau_0,\sigma_i\in\F_l,\quad \tau_1,\rho_i\in\F_{l+1}.
\]
Thus by using two-component boson-fermion correspondence\cite{Kac2003,Kac2023}, we can obtain the reduction of the mToda hierarchy, which is the two-component generalization of the cmKP hierarchy, called the generalized bigraded mToda (GBMT) hierarchy in what follows.

In this paper, there are the following main results:

{\bf Main result 1.}
The reduction of the mToda hierarchy is given by 
\begin{equation}\label{eq:consmToda1}
    L_1(n)^M\cdot\iota_{\La^{-1}}\Delta^{-1}=L_2(n)^N\cdot\iota_{\La}\Delta^{-1}+\sum_{l\in\Z}\sum_{i=1}^{m}q_{i,n}\La^lr_{i,n+1},\qquad M,N\geq1,
\end{equation}
where $\iota_{\La^{-1}}\Delta^{-1}=\sum_{j\geq1}\La^j$, $\iota_{\La}\Delta^{-1}=-\sum_{j\geq0}\La^j$. Further, it can be found that \eqref{eq:consmToda1} is equivalent to 
\begin{equation}
    L_1(n)^M=L_2(n)^N+\sum_{l\in\Z}\sum_{i=1}^{m}q_{i,n}\La^lr_{i,n+1}\Delta,\qquad (L_1(n)^M+L_2(n)^N)(1)=0.\label{eq:consmToda}
\end{equation}
Thus
\begin{equation}\label{eq:consmToda2}
    \begin{aligned}
        &L_1(n)^M=(L_1(n)^M)_{\Delta,\geq1}+(L_2(n)^N)_{\Delta^*,\geq1}+\sum_{i=1}^{m}q_{i,n}\cdot\iota_{\La^{-1}}\Delta^{-1}\cdot r_{i,n+1}\cdot\Delta,\\
        &L_2(n)^N=(L_1(n)^M)_{\Delta,\geq1}+(L_2(n)^N)_{\Delta^*,\geq1}+\sum_{i=1}^{m}q_{i,n}\cdot\iota_{\La}\Delta^{-1}\cdot r_{i,n+1}\cdot\Delta,
    \end{aligned}
\end{equation}
(please see Corollary \ref{coro:consL} and \ref{coro:conL'} in Section \ref{Sect:Lax} for details). Notice that this mToda reduction can be viewed as the generalization of the bigraded mToda hierarchy \eqref{eq:constmBTH}. Thus, we call \eqref{eq:consmToda1} or \eqref{eq:consmToda} the GBMT hierarchy.

{\bf Main result 2.}
The GBMT hierarchy is related by Miura links $\LL_aT_b=T_bL_a~(a,b=1,2)$ with the generalized bigraded Toda (GBT) hierarchy \cite{LiuYue2024} defined by the following constraint on the Toda Lax operators $\LL_1=\Lambda +\sum_{i\geq0}a_i\La^{-i}$ and $\LL_2=b\Lambda^{-1}+\sum_{i\geq0}b_i\La^{i}$ satisfying
\begin{equation}
    \LL_1(n)^M=\LL_2(n)^{N}+\sum_{j\in \mathbb Z}\sum_{i=1}^{m}\tilde{q}_{i,n}\Lambda^j\tilde{r}_{i,n},\qquad M,N\geq1,\label{eq:consToda}
\end{equation}
(see Section \ref{Sect:Miura} for the detailed definitions of $T_b$), which coincides with the fact that there are  Miura links \cite{Shaw1997} between the cmKP hierarchy and the constrained KP hierarchy\cite{Cheng&Zhang1994,ChengY1992} $(\tilde{\fL}^k)_{<0}=\sum_{i=1}^m \tilde{\mathfrak{q}}_i\p^{-1}\tilde{\mathfrak{r}}_i$, since the  GBT hierarchy is analogous to the constrained KP hierarchy (please see \cite{LiuYue2024} for details). 

{\bf Main result 3.}
By using the spectral representations of eigenfunctions of the  Toda and mToda hierarchies, we find equivalent formulations of the  GBMT and GBT hierarchies in terms of tau functions, that is, 
\begin{itemize}
    \item The GBT hierarchy: $$D_{M,N}\Delta(\log\tau_n(\bft))=\sum_{i=1}^{m}\tilde{q}_{i,n}(\bft)\tilde{r}_{i,n}(\bft).$$
    \item The GBMT hierarchy: $$D_{M,N}\Delta(\log\tau_{0,n}(\bft))=-\sum_{i=1}^{m}q_{i,n}(\bft)\Delta(r_{i,n}(\bft)),\qquad D_{M,N}\Delta(\log\tau_{1,n}(\bft))=\sum_{i=1}^{m}\Delta(q_{i,n}(\bft))r_{i,n+1}(\bft),$$
\end{itemize}
where $D_{M,N}=\p_{t_M^{(1)}}+\p_{t_N^{(2)}}$, $\tau_n(\bft)$ is the Toda tau function and $(\tau_{0,n}(\bft),\tau_{1,n}(\bft))$ is the  mToda tau pair.

This paper is organized as follows. In Section \ref{Sect:BieqtoLax}, the bilinear equations of the GBMT hierarchy are obtained by the two-component boson-fermion correspondence from the fermionic cmKP hierarchy. Next in Section \ref{Sect:Lax}, we derive the Lax formulations of the GBMT hierarchy from the bilinear equations. While in Section \ref{sect:Lax->Bieq}, we consider the inverse direction, i.e., from the Lax formulations to the bilinear equations. Then in Section \ref{Sect:Miura}, we show the Miura links between the GBMT hierarchy and the GBT hierarchy. Further in Section \ref{sect:ConsInTau}, we consider the equivalent formulations in tau functions for the constraints of the GBMT and GBT hierarchies, respectively. Finally, some conclusions and discussions are given in Section \ref{sect:Conclusions}.

\section{The GBMT hierarchy from bosonization of the fermionic cmKP hierarchy}\label{Sect:BieqtoLax}

In this section, we firstly review some equivalent formulations of the cmKP hierarchy, especially the bilinear equations. Then we convert the bilinear equations of the cmKP hierarchy into the fermionic form. Finally by using the two-component boson-fermion correspondence, we get the bilinear equations of the GBMT hierarchy.

\subsection{Equivalent formulations of the cmKP hierarchy} In this subsection, let us recall three equivalent formulations of the cmKP hierarchy. One can refer to \cite{ChenCheng2019,Cheng2018,WuCheng2022} for more details.

Firstly, the $(k,m)$-cmKP hierarchy is defined by the following system \cite{Cheng2018,WuCheng2022},
\begin{align}
        &\fL^k=(\fL^k)_{\geq 1}+\sum_{j=1}^m \mathfrak{q}_j\p^{-1}\mathfrak{r}_j\p,\label{eq:cmKP1}\\
        &\p_{x_p}\fL=[(\fL^p)_{\geq1},\fL],\label{eq:cmKP2}\\
        &\p_{x_p}\mathfrak{q}_j=(\fL^p)_{\geq1}(\mathfrak{q}_j),\quad \p_{x_p}\mathfrak{r}_j=-(\p(\fL^p)_{\geq1}\p^{-1})^*(\mathfrak{r}_j),\label{eq:cmKP3}
\end{align}
where the Lax operator $\fL=\p+\sum_{j\geq0}v_j(\bfx)\p^{-j}$ with $\bfx=(x_1=x,x_2,\cdots)$, $(\sum_i a_i\p^i)^*=\sum_i(-\p)^ia_i$ and $(\sum_ia_i\p^i)_{\geq1}=\sum_{i\geq1}a_i\p^i$.

Next if we introduce the mKP wave function $w(\bfx,z)=(e^{\al(\bfx)}+\sum_{i\geq1}w_i(\bfx)z^{-i})e^{\xi(\bfx,z)}$ and the mKP adjoint wave function $w^*(\bfx,z)=(e^{\beta(\bfx)}+\sum_{i\geq0}w_i^*(\bfx)z^{-i})z^{-1}e^{-\xi(\bfx,z)}$ satisfying
\begin{alignat*}{2}
    &\fL \big(w(\bfx,z)\big)=zw(\bfx,z),\quad &&(\p\fL\p^{-1})^*\big(w^*(\bfx,z)\big)=zw^*(\bfx,z),\\
    &\p_{x_p}w(\bfx,z)=(\fL^p)_{\geq1}\big(w(\bfx,z)\big),\qquad &&\p_{x_p}w^*(\bfx,z)=-(\p(\fL^p)_{\geq1}\p^{-1})^*\big(w^*(\bfx,z)\big),
\end{alignat*} 
then the cmKP hierarchy \eqref{eq:cmKP1}--\eqref{eq:cmKP3} is equivalent to the following bilinear equations  in terms of the wave functions \cite{ChenCheng2019,WuCheng2022},
\begin{align*}
    &\Res_z z^{k}w(\bfx,z)w^*(\bfx',z)=\sum_{1\leq i\leq m}\mathfrak{q}_i(\bfx)\mathfrak{r}_i(\bfx'),\\
    &\Res_z \Omega\big(w(\bfx,z),\mathfrak{r}_i(\bfx)_x\big) w^*(\bfx',z)=\mathfrak{r}_i(\bfx)-\mathfrak{r}_i(\bfx'),\\
    &\Res_z \hat{\Omega}\big(\mathfrak{q}_i(\bfx)_x,w^*(\bfx,z)\big)w(\bfx',z)=\mathfrak{q}_i(\bfx)-\mathfrak{q}_i(\bfx'),
\end{align*}
where $\Res_z\sum_i a_iz^i=a_{-1}$, $\xi(\bfx,z)=\sum_{j\geq1}x_jz^j$.
Here $\Omega(\mathfrak{q},\mathfrak{r}_x)$ and $\hat{\Omega}(\mathfrak{q}_x,\mathfrak{r})$ denote the mKP squared eigenfunction potentials \cite{WuCheng2022,Cheng2018,Oevel1998} defined by 
\begin{align*}
    &\p_x\Omega(\mathfrak{q},\mathfrak{r}_x)=\mathfrak{q}\mathfrak{r}_x,\quad \p_{x_p}\Omega(\mathfrak{q},\mathfrak{r}_x)=\Res_\p(\p^{-1}\mathfrak{r}_x(\fL^p)_{\geq1}\mathfrak{q}\p^{-1}),\\
    &\p_x\hat{\Omega}(\mathfrak{q}_x,\mathfrak{r})=\mathfrak{q}_x\mathfrak{r},\quad \p_{x_p}\hat{\Omega}(\mathfrak{q}_x,\mathfrak{r})=\Res_\p(\p^{-1}\mathfrak{r}(\fL^p)_{\geq1}\mathfrak{q}_x\p^{-1}),
\end{align*}
where $\mathfrak{q}$ is the mKP eigenfunction and $\mathfrak{r}$ is the mKP adjoint eigenfunction satisfying $\p_{x_p}\mathfrak{q}=(\fL^p)_{\geq1}(\mathfrak{q})$ and $\p_{x_p}\mathfrak{r}=-(\p(\fL^p)_{\geq1}\p^{-1})^*(\mathfrak{r})$.

Finally, note that there exist two tau functions $\tau_0(\bfx)$ and $\tau_1(\bfx)$ such that \cite{Cheng2018}
\begin{align*}
    w(\bfx,z)=\frac{\tau_0(\bfx-[z^{-1}])}{\tau_1(\bfx)}e^{\xi(\bfx,z)},\qquad w^*(\bfx,z)=\frac{\tau_1(\bfx+[z^{-1}])}{z\tau_0(\bfx)}e^{-\xi(\bfx,z)},
\end{align*}
where $[z^{-1}]=(z^{-1},z^{-2}/2,\cdots)$. Then by using \cite{WuCheng2022,Cheng2018}
\begin{align*}
    &\Omega\big(w(\bfx,z),\mathfrak{r}_i(\bfx)_x\big)=w(\bfx,z)\big(\mathfrak{r}_i(\bfx)-\mathfrak{r}_i(\bfx-[z^{-1}])\big),\\
    &\hat{\Omega}\big(\mathfrak{q}_i(\bfx)_x,w^*(\bfx,z)\big)=w^*(\bfx,z)\big(\mathfrak{q}_i(\bfx)-\mathfrak{q}_i(\bfx+[z^{-1}])\big),
\end{align*} 
we can obtain the bilinear equations in terms of tau functions \cite{WuCheng2022},
\begin{align}
    &\Res_{z} z^{k-1}\tau_0(\bfx-[z^{-1}])\tau_1(\bfx'+[z^{-1}])e^{\xi(\bfx-\bfx',z)}=\sum_{1\leq i\leq m}\rho_i(\bfx)\sigma_i(\bfx'),\label{eq:consmKPBiEq1}\\
    &\Res_z z^{-1}\tau_0(\bfx-[z^{-1}])\rho_i(\bfx'+[z^{-1}])e^{\xi(\bfx-\bfx',z)}=\tau_0(\bfx')\rho_i(\bfx),\\
    &\Res_z z^{-1}\sigma_i(\bfx-[z^{-1}])\tau_{1}(\bfx'+[z^{-1}])e^{\xi(\bfx-\bfx',z)}=\sigma_i(\bfx')\tau_1(\bfx),\\
    &\Res_z z^{-1}\tau_{0}(\bfx-[z^{-1}])\tau_{1}(\bfx'+[z^{-1}])e^{\xi(\bfx-\bfx',z)}=\tau_{0}(\bfx')\tau_1(\bfx),\label{eq:consmKPBiEq3}
\end{align}
where $\rho_i(\bfx)=q_i(\bfx)\tau_1(\bfx)$ and $\sigma_i(\bfx)=r_i(\bfx)\tau_0(\bfx)$.

\subsection{Bosonization for the fermionic cmKP hierarchy: the GBMT hierarchy}
In this subsection, we will convert the $(k,m)$-cmKP hierarchy into the fermionic form, and then by using two-component boson-fermion correspondence to this fermionic form, we will get the bilinear equations of the GBMT hierarchy.

For this, let us introduce the charged free fermions $\psi_j^\pm~(j\in\Z+1/2)$ (see, e.g., \cite{JimboMiwa1983,JM2000,Kac2023}) satisfying
\[
    \psi^\la_i\psi^\mu_j+\psi^\mu_j\psi^\la_i=\delta_{\la,-\mu}\delta_{i+j,0},\quad \la,\mu\in\{+,-\},
\]
then the Fock space $\F$ is the vector space spanned by $\psi^+_{i_1}\cdots\psi^+_{i_r}\psi^-_{j_1}\cdots\psi^-_{j_s}\ket{0}$ for $i_1<\cdots<i_r<0$ and $j_1<\cdots<j_s<0$, where the vacuum $\ket{0}$ satisfies
$
    \psi_n^\pm\ket{0}=0 ~ (n> 0).
$
If we define the charge of $\psi_j^\pm$ to be $\pm1$, then $\psi^+_{i_1}\cdots\psi^+_{i_r}\psi^-_{j_1}\cdots\psi^-_{j_s}\ket{0}$ has the charge $r-s$. Thus if denote $\F_l$ to be the subspace of $\F$ having the charge $l$, then $\F=\bigoplus_{l\in\Z}\F_l$.

There are many different bosonizations for $\F$ (see, e.g., \cite{JimboMiwa1983,JM2000,Kac2023}). Here in this paper we are more interested in the following two bosonizations of $\F$.
\begin{itemize}
    \item $\hat{\mathfrak{b}}:\F\xrightarrow{\sim}\C[q^{\pm1};\bfx=(x_1,x_2\cdots)]$, which is uniquely defined by $\hat{\mathfrak{b}}(\ket{0})=1$ and 
    \[
        \hat{\mathfrak{b}}\psi^\pm(z)\hat{\mathfrak{b}}^{-1}=q^{\pm1}z^{\pm q\p_q}e^{\pm\xi(\bfx,z)}e^{\mp\xi(\tilde{\p},z^{-1})},
    \]
    where $q$ commutes with all $x_i$, $\psi^\pm(z)=\sum_{j\in\Z+1/2}\psi_j^\pm z^{-j-\frac12}$ and $\tilde{\p}=(\p_{x_1},\p_{x_2}/2,\cdots)$.
    \item $\mathfrak{b}:\F\xrightarrow{\sim}\C[Q_1^{\pm1},Q_2^{\pm1};\bft=(\bft^{(1)},\bft^{(2)})\mid\bft^{(a)}=(t^{(a)}_1,t^{(a)}_2\cdots),a=1,2]$, which is uniquely determined by $\mathfrak{b}(\ket{0})=1$ and 
    \begin{equation}\label{bosoForToda}
        \mathfrak{b}\psi^{\pm(a)}(z)\mathfrak{b}^{-1}=Q_a^{\pm 1} z^{\pm Q_a\partial_{Q_a}}e^{\pm \xi(\bft^{(a)},z)}e^{\mp\xi(\tilde{\partial}_a,z^{-1})},
    \end{equation}
    where $Q_a$ satisfies $Q_1Q_2=-Q_2Q_1$ and commutes with all $t_n^{(i)}$, $\psi^{\pm(a)}(z)=\sum_{j\in\Z+1/2}\psi_j^{\pm(a)} z^{-j-\frac12}$ and two-component fermions $\psi_j^{\pm(a)}$ are defined by
    \begin{alignat*}{2}
        &\psi_{Mi+p+1/2}^{+(1)}=\psi_{(M+N)i+p+1/2}^{+},\quad &&\psi_{Mi-p-1/2}^{-(1)}=\psi_{(M+N)i-p-1/2}^{-}, \quad 0\leq p\leq M-1, \\
        &\psi_{Ni+q+1/2}^{+(2)}=\psi_{(M+N)i+M+q+1/2}^{+},\quad &&\psi_{Ni-q-1/2}^{-(2)}=\psi_{(M+N)i-M-q-1/2}^{-},\quad 0\leq q\leq N-1.
    \end{alignat*}
    Here $k=M+N$, and $\mathfrak{b}$ is called the two-component boson-fermion correspondence.
\end{itemize}

By applying the inverse of $\hat{\mathfrak{b}}$ to the cmKP hierarchy \eqref{eq:consmKPBiEq1}--\eqref{eq:consmKPBiEq3}, we can get 
\begin{align}
        &S_k(\tau_0\otimes\tau_1)=\sum_{1\leq i\leq m}\rho_i\otimes\sigma_i,\label{eq:fermicmKP1}\\
        &S_0(\tau_0\otimes\rho_i)=\rho_i\otimes\tau_0,\\
        &S_0(\sigma_i\otimes\tau_1)=\tau_1\otimes\sigma_i,\qquad \sigma_i,\tau_0\in\F_l,\\
        &S_0(\tau_0\otimes\tau_1)=\tau_1\otimes\tau_0,\qquad \rho_i,\tau_1\in\F_{l+1},\label{eq:fermicmKP4}
\end{align}
where $S_k=\sum_{j\in\Z+1/2}\psi^+_{j}\otimes\psi^-_{-j+k}$. This is the fermionic form of the cmKP hierarchy. Notice that the mKP and mToda hierarchies have the same fermionic form \eqref{eq:fermimKP}. Thus we believe that this fermionic form can also correspond to a reduction of the mToda hierarchy, which is analogous to the cmKP hierarchy. 

Let $S_l^{(a)}=\sum_{j\in\Z+1/2}\psi^{+(a)}_j\otimes\psi^{-(a)}_{-j+l}$, then 
\[
    S_0=S_0^{(1)}+S_0^{(2)},\qquad S_k=S_M^{(1)}+S_N^{(2)}.
\]
If we define 
\be\label{Todatau}
    \mathfrak{b}(\tau)=\sum_{n\in\Z}(-1)^{\frac{n(n-1)}{2}}Q_1^{n+l}Q_2^{-n}\tau_{n}(\mathbf{t}),
\ee
then applying the bosonization $\mathfrak{b}$ to the above fermionic equations \eqref{eq:fermicmKP1}--\eqref{eq:fermicmKP4}, we can obtain
\begin{align}
    &\begin{aligned}\label{cmToda1}
        &\oint_{C_\infty}\frac{dz}{2\pi\textnormal{\bf i}z}z^{M+n-n'}\tau_{0,n}(\bft-[z^{-1}]_1)\tau_{1,n'}(\bft'+[z^{-1}]_1)e^{\xi(\bft^{(1)}-\bft'^{(1)},z)}\\
        &\quad +\oint_{C_0}\frac{dz}{2\pi\textnormal{\bf i}z}z^{N+n-n'+1}\tau_{0,n+1}(\bft-[z]_2)\tau_{1,n'-1}(\bft'+[z]_2)e^{\xi(\bft^{(2)}-\bft'^{(2)},z^{-1})}=\sum_{i=1}^{m}\rho_{i,n}(\bft)\sigma_{i,n'}(\bft'),
    \end{aligned}\\
    &\begin{aligned}\label{cmToda2}
        &\oint_{C_\infty}\frac{dz}{2\pi\textnormal{\bf i}z}z^{n-n'}\tau_{0,n}(\bft-[z^{-1}]_1)\rho_{i,n'}(\bft'+[z^{-1}]_1)e^{\xi(\bft^{(1)}-\bft'^{(1)},z)}\\
        &\quad +\oint_{C_0}\frac{dz}{2\pi\textnormal{\bf i}z}z^{n-n'+1}\tau_{0,n+1}(\bft-[z]_2)\rho_{i,n'-1}(\bft'+[z]_2)e^{\xi(\bft^{(2)}-\bft'^{(2)},z^{-1})}=\rho_{i,n}(\bft)\tau_{0,n'}(\bft'),
    \end{aligned}\\
    &\begin{aligned}\label{cmToda3}
        &\oint_{C_\infty}\frac{dz}{2\pi\textnormal{\bf i}z}z^{n-n'}\sigma_{i,n}(\bft-[z^{-1}]_1)\tau_{1,n'}(\bft'+[z^{-1}]_1)e^{\xi(\bft^{(1)}-\bft'^{(1)},z)}\\
        &\quad +\oint_{C_0}\frac{dz}{2\pi\textnormal{\bf i}z}z^{n-n'+1}\sigma_{i,n+1}(\bft-[z]_2)\tau_{1,n'-1}(\bft'+[z]_2)e^{\xi(\bft^{(2)}-\bft'^{(2)},z^{-1})}=\tau_{1,n}(\bft)\sigma_{i,n'}(\bft'),
    \end{aligned}\\
    &\begin{aligned}\label{mToda}
        &\oint_{C_\infty}\frac{dz}{2\pi\textnormal{\bf i}z}z^{n-n'}\tau_{0,n}(\bft-[z^{-1}]_1)\tau_{1,n'}(\bft'+[z^{-1}]_1)e^{\xi(\bft^{(1)}-\bft'^{(1)},z)}\\
        &\quad +\oint_{C_0}\frac{dz}{2\pi\textnormal{\bf i}z}z^{n-n'+1}\tau_{0,n+1}(\bft-[z]_2)\tau_{1,n'-1}(\bft'+[z]_2)e^{\xi(\bft^{(2)}-\bft'^{(2)},z^{-1})}=\tau_{1,n}(\bft)\tau_{0,n'}(\bft'),
    \end{aligned}
\end{align}
where $\textnormal{\bf i}=\sqrt{-1}$ and $\bft-[z^{-1}]_a=(\bft^{(1)},\bft^{(2)})|_{\bft^{(a)}\mapsto\bft^{(a)}-[z^{-1}]}$. Here $C_\infty$ and $C_0$ denote the circle around $z=\infty$ and $z=0$ with anticlockwise, respectively. Equation \eqref{cmToda1} is just the two-component generalization of the cmKP hierarchy, that is, a new reduction of the mToda hierarchy, which is a generalization of the bigraded mToda hierarchy. While \eqref{cmToda2}--\eqref{mToda} imply that $(\tau_{0,n},\tau_{1,n}),(\tau_{0,n},\rho_{i,n})$ and $(\sigma_{i,n},\tau_{1,n})$ are the mToda tau pairs \cite{RuiCheng2024}. Here we call \eqref{cmToda1}--\eqref{mToda} to be the {\bf generalized bigraded mToda (GBMT)} hierarchy.

\begin{example}
    The bilinear equations \eqref{cmToda1}--\eqref{cmToda3} for $M=N=1$ imply
    \begin{equation}
        2\sum_{i=1}^{m}\sigma_{i,n}\rho_{i,n+1}=\big(D^{(1)}_2+(D^{(1)}_1)^2\big)\tau_{1,n}\cdot\tau_{0,n+1},
    \end{equation}
    and
    \begin{align}
        &\tau_{0,n}\rho_{i,n+1}=D^{(1)}_1\rho_{i,n}\cdot\tau_{0,n+1},\quad \sigma_{i,n}\tau_{1,n+1}=D^{(1)}_1\tau_{1,n}\cdot\sigma_{i,n+1},\quad \tau_{0,n}\tau_{1,n+1}=D^{(1)}_1\tau_{1,n}\cdot\tau_{0,n+1},\nonumber
    \end{align}
    where the Hirota derivative \cite{Hirota2004} is defined by
    $
        P(D_\bft)f(\bft)\cdot g(\bft)=P(\p_\bfy)f(\bft+\bfy) g(\bft-\bfy)|_{\bfy=0},
    $
    and $P(D_\bft)$ is a polynomial in $D_\bft=(D^{(1)},D^{(2)})=(D_1^{(1)},D_2^{(1)},\cdots, D_1^{(2)},D_2^{(2)},\cdots)$.
\end{example}

\section{Lax formulation of the GBMT hierarchy}\label{Sect:Lax}

In this section, we derive the following Lax equations of the GBMT hierarchy from the corresponding bilinear equations \eqref{cmToda1}--\eqref{mToda}. 

For this, let us firstly recall some properties of the shift operator $\La$ (see \cite[Section 2]{RuiCheng2024} for more details). The shift operator $\La$ is defined by $\La(f(n,\bft))=f(n+1,\bft)$ for a function $f(n,\bft)\in\mathcal{A}=\{f(n,\bft)\mid  f(n,\bft)\in\C((\bft)), \forall n\in\Z\}$. Here we are interested in the operators in $\mathcal{A}[[\La^{-1},\La]]$, $\mathcal{A}[[\La^{-1}]][\La]$ and $\mathcal{A}[[\La]][\La^{-1}]$. The inverse of the difference operator $\Delta=\La-1$ is denoted as 
\[
    \iota_{\La}\Delta^{-1}=-\sum_{i\geq0}\La^i\in\mathcal{A}[[\La]][\La^{-1}],\qquad \iota_{\La^{-1}}\Delta^{-1}=\sum_{i\geq1}\La^{-i}\in\mathcal{A}[[\La^{-1}]][\La]
\]
For $A=\sum_{j}a_j\La^j\in\mathcal{A}[[\La^{-1}]][\La]$ or $\mathcal{A}[[\La]][\La^{-1}]$, we write $A(1)=\sum_{j}a_j$ if the sum is meaningful. For $F_a=\sum f_{a,j}\La^j\in\mathcal{A}[[\La^{-1},\La]]~(a=1,2)$, its conjugation is $F_a^*=\sum \La^{-j} f_{a,j}$. Thus $\Delta^*=\La^{-1}-1$. The composite of $F_1$ and $F_2$ is denoted as $F_1\cdot F_2$ or $F_1 F_2$, and the action of $F_1$ on $F_2$ is denoted as $F_1(F_2)=\sum_{j,l}f_{1,j}\La^j(f_{2,l})\La^l$. For a subset $P\in\{\geq k, >k, \leq k, <k\}$ of $\Z$, consisting of $n\in\Z$ satisfying $P$,  $(F_a)_{P}=\sum_{j\in P} f_{a,j}\La^j$. Similarly, $(F_a)_{[k]}=f_{a,k}\La^k$. Further $\mathcal{A}[[\La]]=\mathcal{A}[[\Delta]],~\mathcal{A}[[\La^{-1}]]=\mathcal{A}[[\Delta^*]]$. For $A\in\mathcal{A}[[\Delta]]$ or $\mathcal{A}[[\Delta^*]]$, the truncation $A_{\Delta,P}$ or $A_{\Delta^*,P}$ means the same as above, except that $\La$ is replaced by $\Delta$ or $\Delta^*$, respectively. 

Some useful lemmas are given as follows.
\begin{lemma}[\hspace{-0.1pt}{\cite{Adle1999}}]\label{LemForD}
    For $A,B\in\mathcal{A}[[\La^{-1}]][\La]$ or $\mathcal{A}[[\La]][\La^{-1}]$,
    \begin{align*}
        \sum_{l\in\Z}\oint_{C_\infty}\frac{dz}{2\pi\textnormal{\bf i} z}A(n,\La)( z^{\pm n})\cdot B(n+l,\La)(z^{\mp n\mp l})\La^l=A(n,\La)\cdot B^*(n,\La).
    \end{align*}
\end{lemma}
\begin{lemma}[\hspace{-0.1pt}{\cite{RuiCheng2024}}]\label{LemDelta}
    For $A\in\mathcal{A}[[\La^{-1}]][\La]$ and  $B\in\mathcal{A}[[\La]][\La^{-1}]$,
    \begin{align*}
        &\big(A \cdot \iota_{\Lambda^{-1}} \Delta^{*-1}\big)_{\geq 1} \cdot \Delta^* = A_{\geq 1} - A_{\geq 1}(1)= A_{\Delta, \geq 1}, \\
        &\big(A \cdot \iota_{\Lambda^{-1}} \Delta^{*-1}\big)_{\leq 0} \cdot \Delta^* = A_{<0} + A_{\geq 0}(1) = A_{\Delta, \leq 0}, \\
        &\big(B \cdot \iota_{\Lambda} \Delta^{*-1}\big)_{\geq 1} \cdot \Delta^* = B_{>0} + B_{<0}(1)= B_{\Delta^*, \leq 0}, \\
        &\big(B \cdot \iota_{\Lambda} \Delta^{*-1}\big)_{\leq 0} \cdot \Delta^* = B_{<0} - B_{<0}(1) = B_{\Delta^*, \geq 1}.
    \end{align*}
\end{lemma}
\begin{lemma}[\hspace{-0.1pt}{\cite{LiuYue2024}}]\label{Asum}
    For $A\in\mathcal{A}[\La^{-1},\La]$ and $f,g\in\mathcal{A}$,
    \[
        A\cdot \sum_{l\in\Z}f\La^lg=\sum_{l\in\Z}A(f)\La^lg,\qquad  \sum_{l\in\Z}f\La^lg\cdot A^*=\sum_{l\in\Z}f\La^l A(g).
    \]
\end{lemma}
\begin{lemma}[\hspace{-0.1pt}{\cite[Lemma 5]{Yang2024}}]\label{A(1)}
    For $A\in\mathcal{A}[[\La^{-1}]][\La]$ and $B\in\mathcal{A}[[\La]][\La^{-1}]$, then
    $
        A\cdot \iota_{\La^{-1}}\Delta^{-1}-B\cdot \iota_{\La}\Delta^{-1}=0
    $
    is equivalent to 
    $
        A=B,~ A(1)=0.
    $
\end{lemma}
\begin{remark}
    Note that the equation $A=B$ in Lemma \ref{A(1)} implies that $A$ is a finite sum of $\La$, so $A(1)$ makes sense.
\end{remark}
Next notice that \eqref{mToda} is the mToda hierarchy \cite{RuiCheng2024}, hence we can introduce the wave functions $\Psi_a(n,\bft,z)$ and the adjoint wave functions $\Psi_a^*(n,\bft,z)$ as follows,
\begin{align}
    &\Psi_1(n,\bft,z)=\frac{\tau_{0,n}(\bft-[z^{-1}]_1)}{\tau_{1,n}(\bft)}z^ne^{\xi(\bft^{(1)},z)}=\psi_1(n,\bft,z)z^ne^{\xi(\bft^{(1)},z)},\label{Psi1}\\
    &\Psi_1^*(n,\bft,z)=\frac{\tau_{1,n}(\bft+[z^{-1}]_1)}{\tau_{0,n}(\bft)}z^{-n}e^{-\xi(\bft^{(1)},z)}=\psi_1^*(n,\bft,z)z^{-n}e^{-\xi(\bft^{(1)},z)},\label{Psi1*}\\
    &\Psi_2(n,\bft,z)=\frac{\tau_{0,n+1}(\bft-[z]_2)}{\tau_{1,n}(\bft)}z^ne^{\xi(\bft^{(2)},z^{-1})}=\psi_2(n,\bft,z)z^ne^{\xi(\bft^{(2)},z^{-1})}, \label{Psi2}\\
    &\Psi_2^*(n,\bft,z)=\frac{\tau_{1,n-1}(\bft+[z]_2)}{\tau_{0,n}(\bft)}z^{1-n}e^{-\xi(\bft^{(2)},z^{-1})}=\psi_2^*(n,\bft,z)z^{-n}e^{-\xi(\bft^{(2)},z^{-1})},\label{Psi2*}
\end{align}
where $\psi_i$ and $\psi^*_i$ have the following expansions of $z$,
\begin{alignat*}{2}
    &\psi_1(n,\bft,z)=\sum_{k\geq0}\psi_{1,k}(n,\bft)z^{-k},\qquad &&\psi_1^*(n,\bft,z)=\sum_{k\geq0}\psi_{1,k}^*(n,\bft)z^{-k},\\
    &\psi_2(n,\bft,z)=\sum_{k\geq0}\psi_{2,k}(n,\bft)z^k,\qquad &&\psi_2^*(n,\bft,z)=\sum_{k\geq0}\psi_{2,k}^*(n,\bft)z^{k+1}.
\end{alignat*}
Here $\psi_{a,k}(n,\bft)$ and $\psi_{a,k}^*(n,\bft)$ can be expressed by $\tau_{0,n}(\bft)$ and $\tau_{1,n}(\bft)$. In particular, $\psi_{1,0}=\tau_{0,n}/\tau_{1,n},$ $\psi_{1,0}^*=\tau_{1,n}/\tau_{0,n}$ and $\psi_{2,0}=\tau_{0,n+1}/\tau_{1,n},~\psi_{2,0}^*=\tau_{1,n-1}/\tau_{0,n}.$
\begin{lemma}[\hspace{-0.1pt}{\cite{LiuYue2024}}]
    Assume that $(f_n,g_n)$ satisfies the mToda hierarchy, i.e.
    \begin{equation}\label{eq:fg}
        \begin{aligned}
            &\oint_{C_\infty}\frac{dz}{2\pi\textnormal{\bf i}z}z^{n-n'}f_n(\bft-[z^{-1}]_1)g_{n'}(\bft'+[z^{-1}]_1)e^{\xi(\bft^{(1)}-\bft'^{(1)},z)}\\
            &\quad +\oint_{C_0}\frac{dz}{2\pi\textnormal{\bf i}z}z^{n-n'+1}f_{n+1}(\bft-[z]_2)g_{n'-1}(\bft'+[z]_2)e^{\xi(\bft^{(2)}-\bft'^{(2)},z^{-1})}=g_{n}(\bft)f_{n'}(\bft'),
        \end{aligned}
    \end{equation}
    then
    \begin{align}
            &g_{n+1}(\bft+[z^{-1}]_1)f_{n}(\bft)-z\cdot g_{n}(\bft+[z^{-1}]_1)f_{n+1}(\bft)=-z\cdot f_{n+1}(\bft+[z^{-1}]_1)g_{n}(\bft),\label{mToda1}\\ 
            &g_{n}(\bft+[z]_2)f_{n}(\bft)-z\cdot g_{n-1}(\bft+[z]_2)f_{n+1}(\bft)=f_{n}(\bft+[z]_2)g_{n}(\bft).\label{mToda2}
        \end{align}
\end{lemma}
\begin{proof}
    Letting $(n',\bft')=(n+1,\bft-[z^{-1}]_1)$ and $(n,\bft-[z]_2)$ in \eqref{eq:fg}, we can get these identities by direct computations.
\end{proof}

\begin{lemma}\label{lem:DeltaqPsi}
    Given mToda tau pairs $(\tau_{0,n},\tau_{1,n}),(\tau_{0,n},\rho_{i,n})$ and $(\sigma_{i,n},\tau_{1,n})$, if define 
    \begin{equation}
        q_{i,n}(\bft)=\frac{\rho_{i,n}(\bft)}{\tau_{1,n}(\bft)},\qquad r_{i,n}(\bft)=\frac{\sigma_{i,n}(\bft)}{\tau_{0,n}(\bft)},\label{qr}
    \end{equation}
    then functions $q_{i,n},r_{i,n}$ and $\Psi_a,\Psi_a^*$ defined in \eqref{Psi1}--\eqref{Psi2*} and \eqref{qr} satisfy
    \begin{align}
        &\Delta\big(q_{i,n}(\bft+[z^{-1}]_1)\Psi_1^*(n,\bft,z)\big)=q_{i,n}(\bft)\Delta\big(\Psi_1^*(n,\bft,z)\big),\label{eq:qPsi1*}\\
        &\Delta\big(q_{i,n-1}(\bft+[z]_2)\Psi_2^*(n,\bft,z)\big)=q_{i,n}(\bft)\Delta\big(\Psi_2^*(n,\bft,z)\big),\label{eq:qPsi2*}\\
        &\Delta\big(r_{i,n}(\bft-[z^{-1}]_1)\Psi_1(n,\bft,z)\big)=r_{i,n+1}(\bft)\Delta\big(\Psi_1(n,\bft,z)\big),\label{eq:rPsi1}\\
        &\Delta\big(r_{i,n+1}(\bft-[z]_2)\Psi_2(n,\bft,z)\big)=r_{i,n+1}(\bft)\Delta\big(\Psi_2(n,\bft,z)\big).\label{eq:rPsi2}
    \end{align}
\end{lemma}
\begin{proof}
    We give a proof of \eqref{eq:qPsi1*}. Recalling the definitions \eqref{Psi1*} and \eqref{mToda1} for $(f_n,g_n)=(\tau_{0,n},\rho_{i,n})$, we have
    \begin{align*}
        \Delta\big(q_{i,n}(\bft+[z^{-1}]_1)\Psi_1^*(n,\bft,z)\big)
        =\frac{z^{-n}e^{-\xi(\bft^{(1)},z)}}{\tau_{0,n}(\bft)\tau_{0,n+1}(\bft)}\tau_{0,n+1}(\bft+[z^{-1}]_1)\rho_{i,n}(\bft).
    \end{align*}
    Since $\rho_{i,n}=q_{i,n}\tau_{1,n}$, we can obtain the following by using \eqref{mToda1} for $(f_n,g_n)=(\tau_{0,n},\tau_{1,n})$,
    \begin{align*}
        q_{i,n}(\bft)\frac{z^{-n}e^{-\xi(\bft^{(1)},z)}}{\tau_{0,n}(\bft)\tau_{0,n+1}(\bft)}\tau_{0,n+1}(\bft+[z^{-1}]_1)\tau_{1,n}(\bft)=q_{i,n}(\bft)\Delta\big(\Psi_1^*(n,\bft,z)\big).
    \end{align*}
    A similar proof for the remaining equations is available.
\end{proof}
\begin{proposition}
    For $\tau_{0,n},\tau_{1,n},\rho_{i,n}$ and $\sigma_{i,n}$ satisfying \eqref{cmToda1}--\eqref{mToda}, if we define $\Psi_a,\Psi_a^*$ by \eqref{Psi1}--\eqref{Psi2*} and $q_{i,n},r_{i,n}$ by \eqref{qr}, then
    \begin{align}
        &\begin{aligned}\label{Wave3}
            \oint_{C_\infty}\frac{dz}{2\pi\textnormal{\bf i} z}&z^M\Psi_1(n,\bft,z)\Psi_1^*(n',\bft',z)\\
            &+\oint_{C_0}\frac{dz}{2\pi\textnormal{\bf i} z}z^{-N}\Psi_2(n,\bft,z)\Psi_2^*(n',\bft',z)=\sum_{i=1}^{m}q_{i,n}(\bft)r_{i,n'}(\bft'),
        \end{aligned}\\
        &\begin{aligned}\label{Wave}
            \oint_{C_\infty}\frac{dz}{2\pi\textnormal{\bf i} z}&\Psi_1(n,\bft,z)\cdot \iota_{\La}\Delta^{-1}\Big( q_{i,n'}(\bft')\cdot \Delta\big(\Psi_1^*(n',\bft',z)\big)\Big)\\
            &+\oint_{C_0}\frac{dz}{2\pi\textnormal{\bf i} z}\Psi_2(n,\bft,z)\cdot\iota_{\La^{-1}}\Delta^{-1}\Big(q_{i,n'}(\bft')\cdot \Delta\big(\Psi_2^*(n',\bft',z)\big)\Big)=q_{i,n}(\bft),
        \end{aligned}\\
        &\begin{aligned}\label{Wave1'}
            \oint_{C_\infty}\frac{dz}{2\pi\textnormal{\bf i} z}&\iota_{\La^{-1}}\Delta^{-1}\Big(r_{i,n'+1}(\bft')\cdot\Delta\big(\Psi_1(n',\bft',z)\big)\Big)\cdot \Psi_1^*(n,\bft,z)\\
            &+\oint_{C_0}\frac{dz}{2\pi\textnormal{\bf i} z}\iota_{\La}\Delta^{-1}\Big(r_{i,n'+1}(\bft')\cdot\Delta\big(\Psi_2(n',\bft',z)\big)\Big)\cdot\Psi_2^*(n,\bft,z)=r_{i,n}(\bft),
        \end{aligned}\\
        &\begin{aligned}\label{Wave2'}
            \oint_{C_\infty}\frac{dz}{2\pi\textnormal{\bf i} z}\Psi_1(n,\bft,z)\cdot\Psi_1^*(n',\bft',z)
            +\oint_{C_0}\frac{dz}{2\pi\textnormal{\bf i} z}\Psi_2(n,\bft,z)\cdot\Psi_2^*(n',\bft',z)=1.
        \end{aligned}
\end{align}
\end{proposition}
\begin{proof}
    By dividing $\tau_{1,n}(\bft)\tau_{0,n'}(\bft')$ at both sides of \eqref{cmToda1}--\eqref{mToda} and using Lemma \ref{lem:DeltaqPsi}, one can obtain \eqref{Wave3}--\eqref{Wave2'}.
\end{proof}

Note that $a(n)\Delta(b(n))=\Delta(a(n)\cdot b(n))-\Delta(a(n))b(n+1)$, then we have the following corollary.

\begin{corollary}
    The identities \eqref{Wave}--\eqref{Wave2'} are equivalent to the following identities,
    \begin{align}
        &\begin{aligned}\label{Wave'}
            \oint_{C_\infty}&\frac{dz}{2\pi\text{i} z}\iota_{\La}\Delta^{-1}\Big(\Delta\big(q_{i,n}(\bft)\big)\cdot \Psi_1^*(n+1,\bft,z)\Big)\cdot\Psi_1(n',\bft',z)\\
            &+\oint_{C_0}\frac{dz}{2\pi\text{i} z}\iota_{\La^{-1}}\Delta^{-1}\Big(\Delta\big(q_{i,n}(\bft)\big)\cdot\Psi_2^*(n+1,\bft,z)\Big)\cdot\Psi_2(n',\bft',z)=q_{i,n}(\bft)-q_{i,n'}(\bft'),
        \end{aligned}\\
        &\begin{aligned}\label{Wave1''}
            \oint_{C_\infty}&\frac{dz}{2\pi\text{i} z}\iota_{\La^{-1}}\Delta^{-1}\Big(\Delta\big(r_{i,n}(\bft)\big)\cdot\Psi_1(n,\bft,z)\Big)\cdot \Psi_1^*(n',\bft',z)\\
            &+\oint_{C_0}\frac{dz}{2\pi\text{i} z}\iota_{\La}\Delta^{-1}\Big(\Delta\big(r_{i,n}(\bft)\big)\cdot\Psi_2(n,\bft,z)\Big)\cdot\Psi_2^*(n',\bft',z)=r_{i,n}(\bft)-r_{i,n'}(\bft').
        \end{aligned}
    \end{align}
\end{corollary}
We also need to introduce the following wave operators:
\begin{alignat*}{2}
    &S_a(n,\bft,\La)=\psi_a(n,\bft,\La),\qquad &&\tilde{S}_a(n,\bft,\La)=\psi_a^*(n,\bft,\La^{-1}),
\end{alignat*}
then 
\begin{alignat}{2}
    &\Psi_1(n,\bft,z)=e^{\xi(\bft^{(1)},z)}S_1(n,\bft,\La)(z^n),\qquad &&\Psi_1^*(n,\bft,z)=e^{-\xi(\bft^{(1)},z)}\tilde{S}_1(n,\bft,\La)(z^{-n}),\label{W1}\\
    &\Psi_2(n,\bft,\La)=e^{\xi(\bft^{(2)},z^{-1})}S_2(n,\bft,\La)(z^n),\quad  &&\Psi_2^*(n,\bft,z)=e^{-\xi(\bft^{(2)},z^{-1})}\tilde{S}_2(n,\bft,\La)(z^{-n}).\label{W2}
\end{alignat}
Then notice that the bilinear equation \eqref{mToda} is just the mToda hierarchy. Thus it follows from \cite{RuiCheng2024} that there are the following identities,
\begin{alignat}{2}
    &\tilde{S}_1=-\iota_{\La}\Delta^{-1}\cdot (S_1^*)^{-1},\qquad &&\tilde{S}_2=\iota_{\La^{-1}}\Delta^{-1}\cdot (S_2^*)^{-1},\label{ConjS}\\
    &\p_{t_p^{(1)}}S_1=B_{p}^{(1)}\cdot S_1-S_1\La^p,\qquad &&\p_{t_p^{(1)}}S_2=B_{p}^{(1)}\cdot S_2,\nonumber\\
    &\p_{t_p^{(2)}}S_2=B_{p}^{(2)}\cdot S_2-S_2\La^{-p},\qquad &&\p_{t_p^{(2)}}S_1=B_{p}^{(1)}\cdot S_1,\nonumber\\
    &\p_{t_p^{(a)}}\Psi_b=B_{p}^{(a)}(\Psi_b),\qquad &&\p_{t_p^{(a)}}\Psi_b^*=-\big(\Delta^{-1} B_{p}^{(a)*}\Delta\big)(\Psi_b^*),\label{partialW1}
\end{alignat}
where $B_{p}^{(1)}=(S_1\La^pS_1^{-1})_{\Delta,\geq1}$ and $B_{p}^{(2)}=(S_2\La^{-p}S_2^{-1})_{\Delta^*,\geq1}$. 
\begin{remark}\label{RemarkpW}
    Here let us explain $\Delta^{-1}$ in \eqref{partialW1}. If we let $B_p^{(1)}=\sum_{i=1}^{p}u_i\Delta^i$, then $\iota_{\La^{-1}}\Delta^{-1} B_{p}^{(1)*}=-\sum_{i=1}^{p}\La^{-1}(\Delta^*)^{i-1}u_i=\iota_{\La}\Delta^{-1} B_{p}^{(1)*}$ is valid and the same. Similarly, $\iota_{\La^{\pm1}}\Delta^{-1} B_{p}^{(2)*}$ is the same. For brevity we will denote it as $\Delta^{-1}$.
\end{remark}

Now let us derive the constraint for the mToda hierarchy. Recalling the above definitions \eqref{W1} and \eqref{W2}, the equation \eqref{Wave3} for $\bft'=\bft$ can be written as  
\begin{align*}
    \sum_{l\in\Z}\Big(&\oint_{C_\infty}\frac{dz}{2\pi\textnormal{\bf i} z}S_1(n,\bft,\La)\La^M(z^n)\cdot \big(\tilde{S}_1(n+l,\bft,\La)\big)(z^{-n-l})\La^l\\
    &+\oint_{C_0}\frac{dz}{2\pi\textnormal{\bf i} z}S_2(n,\bft,\La)\La^{-N}(z^n)\cdot \big(\tilde{S}_2(n+l,\bft,\La)\big)(z^{-n-l})\La^l\Big)=\sum_{l\in\Z}\sum_{i=1}^{m}q_{i,n}(\bft)r_{i,n+l}(\bft)\La^l.
\end{align*}
Then using Lemma \ref{LemForD} and \eqref{ConjS}, we have
\begin{equation}\label{ConstS1S2}
    S_1\cdot \La^M\cdot S_1^{-1}\iota_{\La^{-1}}\Delta^{-1}=S_2\cdot \La^{-N}\cdot S_2^{-1}\iota_{\La}\Delta^{-1}+\sum_{l\in\Z}\sum_{i=1}^{m}q_{i,n}\La^lr_{i,n+1}.
\end{equation}
Moreover, applying $\p_{t_p^{(a)}}$ and $\p_{t_p'^{(a)}}$ to both sides of \eqref{Wave} and \eqref{Wave1'} respectively, and using \eqref{partialW1}, we have 
\[
    \p_{t_p^{(a)}}(q_{i,n})=B_p^{(a)}(q_{i,n}),\qquad \p_{t_p^{(a)}}(r_{i,n})=-\big(\Delta^{-1} B_{p}^{(a)*}\Delta\big)(r_{i,n}),
\]
which imply that $q_{i,n}$ is the mToda eigenfunction and $r_{i,n}$ is the mToda adjoint eigenfunction \cite{RuiCheng2024}. So if we define 
\begin{equation}\label{Laxop}
    \begin{aligned}
        &L_1(n,\bft,\La)=S_1(n,\bft,\La)\cdot \La\cdot S_1^{-1}(n,\bft,\La)=u(n,\bft)\La+\sum_{j\geq0}u_j(n,\bft)\La^{-j},\\
        &L_2(n,\bft,\La)=S_2(n,\bft,\La)\cdot \La^{-1}\cdot S_2^{-1}(n,\bft,\La)=\bar{u}(n,\bft)\La^{-1}+\sum_{j\geq0}\bar{u}_j(n,\bft)\La^j.
    \end{aligned}
\end{equation}
Then we obtain the following theorem.
\begin{theorem}\label{ThLax}
    The functions $\{q_{i,n},r_{i,n}\}_{1\leq i \leq m}$ defined by \eqref{qr} and the operators $L_a$ defined in \eqref{Laxop} satisfy
    \begin{align}
        &L_1(n)^M\cdot\iota_{\La^{-1}}\Delta^{-1}=L_2(n)^N\cdot\iota_{\La}\Delta^{-1}+\sum_{l\in\Z}\sum_{i=1}^{m}q_{i,n}\La^lr_{i,n+1},\label{ConsL}\\
        &\p_{t_p^{(a)}}L_b=[B_p^{(a)},L_b],\qquad a,b=1,2,\label{Laxt1}\\
        &\p_{t_p^{(a)}}(q_{i,n})=B_p^{(a)}(n)(q_{i,n}),\qquad \p_{t_p^{(a)}}(r_{i,n})=-\big(\Delta^{-1} B_{p}^{(a)*}(n)\Delta\big)(r_{i,n}),\label{partialqr'}
    \end{align}
    where $B_p^{(1)}(n)=(L_1^p(n))_{\Delta,\geq1}$ and $B_p^{(2)}(n)=(L_2^p(n))_{\Delta^*,\geq1}$.
\end{theorem}
Moreover, by \eqref{W1}--\eqref{W2}, the following auxiliary equations hold.
\begin{corollary}\label{coro:LPsi}
    For $a,b\in\{1,2\}$, we have 
    \begin{alignat*}{2}  
        &L_1(\Psi_1)=z\Psi_1,\qquad &&\big(\iota_{\La}\Delta^{-1}L_1^*\Delta\big)(\Psi_1^*)=z\Psi_1^*,\\
        &L_2(\Psi_2)=z^{-1}\Psi_2,\qquad &&\big(\iota_{\La^{-1}}\Delta^{-1}L_2^*\Delta\big)(\Psi_2^*)=z^{-1}\Psi_2^*,\\
        &\p_{t_p^{(a)}}(\Psi_b)=B_p^{(a)}(\Psi_b),\qquad &&\p_{t_p^{(a)}}(\Psi_b^*)=-\big(\Delta^{-1}B_p^{(a)*}\Delta\big)(\Psi_b^*).
    \end{alignat*}
\end{corollary}
\begin{corollary}\label{coro:consL}
    The GBMT constraint \eqref{ConsL} is equivalent to 
    \begin{equation}
        L_1(n)^M=L_2(n)^N+\sum_{l\in\Z}\sum_{i=1}^{m}q_{i,n}\La^lr_{i,n+1}\Delta,\qquad \big(L_1^M(n)+L_2^N(n)\big)(1)=0.\label{eq:ConsL1}
    \end{equation}
\end{corollary}
\begin{proof}
    Firstly let us prove $\eqref{ConsL}\Rightarrow\eqref{eq:ConsL1}$. For this, denote
    \begin{align*}
        &A=L_1(n)^M-\sum_{i=1}^{m}q_{i,n}\cdot\iota_{\La^{-1}}\Delta^{-1}\cdot r_{i,n+1}\cdot\Delta,\qquad B=L_2(n)^N-\sum_{i=1}^{m}q_{i,n}\cdot\iota_{\La}\Delta^{-1}\cdot r_{i,n+1}\cdot \Delta,
    \end{align*}
    Notice that $A\in\mathcal{A}[[\La^{-1}]][\La],~B\in\mathcal{A}[[\La]][\La^{-1}]$ and we know $A\cdot \iota_{\La^{-1}}\Delta^{-1}=B\cdot \iota_{\La}\Delta^{-1}$ by \eqref{ConsL}. Thus from Lemma \ref{A(1)}, we get $A=B$, that is the first relation in \eqref{eq:ConsL1}, and $A(1)=B(1)=0$, which means $L_1^M(1)=L_2^N(1)=0$, since $\big(\iota_{\La^{\pm1}}\Delta^{-1} r_{i,n+1} \Delta\big)(1)=\big(r_{i,n}-\iota_{\La^{\pm1}}\Delta^{-1}\cdot \Delta(r_{i,n})\big)(1)=0$.

    Conversely, let us show $\eqref{eq:ConsL1}\Rightarrow\eqref{ConsL}$. From the first equation in \eqref{eq:ConsL1}, we get $A=B$. Further, there are also $A(1)=L_1^M(1)$ and $B(1)=L_2^N(1)$, so $(A+B)(1)=0$. Therefore according to Lemma \ref{A(1)}, we have 
    $
        A\cdot \iota_{\La^{-1}}\Delta^{-1}=B\cdot \iota_{\La}\Delta^{-1},
    $
    which is \eqref{ConsL}.
\end{proof}

\begin{corollary}\label{coro:conL'}
    The GBMT constraint \eqref{eq:ConsL1} is also equivalent to 
    \begin{align}
        &L_1(n)^M=B_M^{(1)}(n)+B_N^{(2)}(n)+\sum_{i=1}^{m}q_{i,n}\cdot\iota_{\La^{-1}}\Delta^{-1}\cdot r_{i,n+1}\cdot\Delta,\label{eq:consL1}  \\
        &L_2(n)^N=B_M^{(1)}(n)+B_N^{(2)}(n)+\sum_{i=1}^{m}q_{i,n}\cdot\iota_{\La}\Delta^{-1}\cdot r_{i,n+1}\cdot\Delta. \label{eq:consL2}
    \end{align}
\end{corollary}
\begin{proof}
    Firstly let us prove \eqref{eq:consL1}--\eqref{eq:consL2} $\Rightarrow\eqref{eq:ConsL1}$. The first equation for \eqref{eq:ConsL1}
    can be obtained by subtracting \eqref{eq:consL1} from \eqref{eq:consL2}. While notice that $B_p^{(a)}(n)(1)=0$ and $\big(\iota_{\La^{\pm1}}\Delta^{-1} r_{i,n+1} \Delta\big)(1)=\big(r_{i,n}-\iota_{\La^{\pm1}}\Delta^{-1}\cdot \Delta(r_{i,n})\big)(1)=0$, we have $L_1^M(1)=L_2^N(1)=0$.

    Next let us prove \eqref{eq:ConsL1} $\Rightarrow$ \eqref{eq:consL1}--\eqref{eq:consL2}. From the first equation of \eqref{eq:ConsL1}, we obtain 
    \begin{align}
        &(L_1^M)_{\Delta,\geq1}=(L_2^N)_{\Delta,\geq1}-\sum_{i=1}^{m}q_{i,n}\cdot\iota_{\La}\Delta^{-1}\cdot r_{i,n+1}\cdot\Delta,\label{eq:L1>0}\\
        &(L_1^M)_{\Delta^*,\geq1}=(L_2^N)_{\Delta^*,\geq1}+\sum_{i=1}^{m}q_{i,n}\cdot\iota_{\La^{-1}}\Delta^{-1}\cdot r_{i,n+1}\cdot\Delta,\label{eq:L1>0*}
    \end{align}
    where we have used $(A\Delta)_{\Delta,\geq1}=A_{\Delta,\geq0}\Delta=A_{\geq0}\Delta$ and $(A\Delta^*)_{\Delta^*,\geq1}=A_{\leq0}\Delta^*$. By the first relation in \eqref{eq:ConsL1}, we know $L_1^M(1)=L_2^N(1)$. So using the second relation in \eqref{eq:ConsL1}, we get $L_1^M(1)=L_2^N(1)=0$. Further by $L_1^M(1)=0$, we get $(L_1^M)_{\Delta,\leq0}=(L_1^M)_{<0}+(L_1^M)_{\geq0}(1)=(L_1^M)_{\leq0}-(L_1^M)_{\leq0}(1)=(L_1^M)_{\Delta^*,\geq1}$. Thus
    \[
        L_1^M=(L_1^M)_{\Delta,\geq1}+(L_1^M)_{\Delta,\leq0}=(L_1^M)_{\Delta,\geq1}+(L_1^M)_{\Delta^*,\geq1}.
    \]
    After taking \eqref{eq:L1>0*} into above relation, we can obtain \eqref{eq:consL1}. Similarly, we can prove \eqref{eq:consL2}.
\end{proof}

\begin{remark}\label{Rmk:L1L2}
    If we let $B_M^{(1)}(n)=\sum_{l=1}^{M}b_l^{(M)}\Delta^l$ and $B_N^{(2)}(n)=\sum_{j=1}^{N}\bar{b}_j^{(N)}(\Delta^*)^j$, then \eqref{eq:consL1}--\eqref{eq:consL2} imply
    \begin{align*}
        &L_1^M=\sum_{l=1}^{M}b_l^{(M)}\Delta^l+\sum_{j=1}^{N}\bar{b}_j^{(N)}(\Delta^*)^j+\sum_{i=1}^{m}q_{i,n}\cdot\iota_{\La^{-1}}\Delta^{-1}\cdot r_{i,n+1}\cdot\Delta,\\
        &L_2^N=\sum_{l=1}^{M}b_l^{(M)}\Delta^l+\sum_{j=1}^{N}\bar{b}_j^{(N)}(\Delta^*)^j+\sum_{i=1}^{m}q_{i,n}\cdot\iota_{\La}\Delta^{-1}\cdot r_{i,n+1}\cdot\Delta,
    \end{align*}
    which implies that there are only $M+N+2m$ independent functions, that is $q_{i,n},r_{i,n}$ and $b_l^{(M)},\bar{b}_j^{(N)}$ for $1\leq i\leq m,~ 1\leq l\leq M,~ 1\leq j\leq N$.
\end{remark}

\begin{example}
    When $M=N=m=1$, $B_1^{(1)}=b_1\Delta,~B_1^{(2)}=\bar{b}_1\Delta^*$, thus 
    \[
        L_1=b_1\Delta+\bar{b}_1\Delta^*+q_n\cdot\iota_{\La^{-1}}\Delta^{-1}\cdot r_{n+1}\cdot\Delta,\qquad L_2=b_1\Delta+\bar{b}_1\Delta^*+q_n\cdot\iota_{\La}\Delta^{-1}\cdot r_{n+1}\cdot\Delta.
    \]
    There are only four independent functions $q_n,r_n,$ $b_{1}$ and $\bar{b}_1$. Then the Lax equations \eqref{Laxt1} and \eqref{partialqr'} imply
    \begin{alignat*}{2}
        &\p_{t_1^{(1)}}b_1=b_1\cdot\Delta(q_nr_n-\bar{b}_1),\qquad &&  \p_{t_1^{(2)}}b_1=b_1\cdot\Delta(\bar{b}_{1}),\\
        &\p_{t_1^{(1)}}q_n=b_{1}\cdot\Delta(q_n),\qquad &&  \p_{t_1^{(2)}}q_n=-\bar{b}_{1}\cdot\Delta(q_{n-1}),\\
        &\p_{t_1^{(1)}}r_n=b_{1}(n-1)\cdot\Delta(r_{n-1}),\qquad &&\p_{t_1^{(2)}}r_n=-\bar{b}_{1}\cdot\Delta(r_n),\\
        &\p_{t_1^{(1)}}\bar{b}_1=\bar{b}_1\cdot\Delta(b_1)+q_n\cdot\Delta(r_{n-1})(b_1(n-1)&&-b_1(n+1)),\\
        &\p_{t_1^{(2)}}\bar{b}_1=b_1\cdot\big(\Delta(q_n)r_n+\bar{b}_1(n-1)-\Delta(\bar{b}_1)\big)&&+b_{1}(n-1)\cdot\big(q_n\Delta(r_{n-1})-\bar{b}_1\big).\\
    \end{alignat*}
\end{example}

\section{Bilinear equations from the Lax formulations}\label{sect:Lax->Bieq}

In Section \ref{Sect:Lax}, we have obtained the Lax formulation of the GBMT hierarchy from the corresponding bilinear equations. In this section, we consider the inverse direction, that is to derive the GBMT bilinear equations from the Lax formulation, so that we can get the equivalence of the bilinear equations and Lax formulation in the GBMT case.

Firstly, the Lax formulation of the GBMT hierarchy is given by Lax operators $L_1=u_{-1}\La+\sum_{j\geq0}u_j\La^{-j}$ and $L_2=\bar{u}_{-1}\La^{-1}+\sum_{j\geq0}\bar{u}_j\La^{j}~ (u_{-1},\bar{u}_{-1}\neq0)$ and unknown functions $\{q_{i,n},r_{i,n}\}_{i=1}^m$, satisfying
\begin{alignat}{2}
    &L_1^M\iota_{\La^{-1}}\Delta^{-1}=L_2^N\iota_{\La}\Delta^{-1}+&&\sum_{l\in\Z}\sum_{i=1}^{m}q_{i,n}\La^lr_{i,n+1},\label{ConsL'}\\
    &\p_{t_p^{(a)}}L_1=[B_p^{(a)},L_1],\qquad &&\p_{t_p^{(a)}}L_2=[B_p^{(a)},L_2],\label{Laxeq}\\
    &\p_{t_p^{(a)}}(q_{i,n})=B_p^{(a)}(q_{i,n}),\qquad &&\p_{t_p^{(a)}}(r_{i,n})=-\big(\Delta^{-1} B_{p}^{(a)*}\Delta\big)(r_{i,n}),\label{partialqr}
\end{alignat}
where $B_p^{(1)}=(L_1^p)_{\Delta,\geq1}$ and $B_p^{(2)}=(L_2^p)_{\Delta^*,\geq1}$.

Let us first point out that the above definition is well-defined. 
\begin{lemma}
    The Lax formulation of the GBMT hierarchy \eqref{ConsL'}--\eqref{partialqr} is well-defined, i.e. applying $\p_{t_p^{(a)}}$ to both sides of \eqref{ConsL'}, and using \eqref{Laxeq}--\eqref{partialqr}, the obtained equation still holds, that is
    \begin{align*}
        &[B_p^{(a)},L_1^M]\iota_{\La^{-1}}\Delta^{-1}-[B_p^{(a)},L_2^N]\iota_{\La}\Delta^{-1}\\
        =&\sum_{l\in\Z}\sum_{i=1}^{m}\Big(B_p^{(a)}(q_{i,n})\cdot\La^l\cdot r_{i,n+1}-q_{i,n}\cdot\La^l\cdot\big(\Delta^{-1} B_{p}^{(a)*}\Delta\big)(r_{i,n})\Big).
    \end{align*}
\end{lemma}
\begin{proof}
    Multiplying $B_p^{(a)}$ to both sides of \eqref{ConsL'} and using Lemma \ref{Asum}, we have 
    \[
        B_p^{(a)}\cdot L_1^M\cdot\iota_{\La^{-1}}\Delta^{-1}=B_p^{(a)}\cdot L_2^N\cdot\iota_{\La}\Delta^{-1}+\sum_{l\in\Z}\sum_{i=1}^{m}B_p^{(a)}(q_{i,n})\cdot\La^l\cdot r_{i,n+1}.
    \]
    On the other hand, by Corollary \ref{coro:consL}, we know $L_1(n)^M=L_2(n)^N+\sum_{l\in\Z}\sum_{i=1}^{m}q_{i,n}\La^lr_{i,n+1}\Delta$. While recalling Remark \ref{RemarkpW}, we get $B_p^{(a)}\cdot\iota_{\La^{-1}}\Delta^{-1}=B_p^{(a)}\cdot\iota_{\La}\Delta^{-1}$, thus again using Lemma \ref{Asum},
    \[
        L_1^M\cdot B_p^{(a)}\cdot\iota_{\La^{-1}}\Delta^{-1}=L_2^N\cdot B_p^{(a)}\cdot\iota_{\La}\Delta^{-1}+\sum_{l\in\Z}\sum_{i=1}^{m}q_{i,n}\cdot\La^l\cdot\big(\Delta^{-1}\cdot B_p^{(a)*}(n+1)\cdot\Delta\big)(r_{i,n+1}).
    \]
    Notice that $\Delta^{-1}\cdot B_p^{(a)*}(n+1)\cdot\Delta=\big(\Delta \cdot B_p^{(a)}(n)\cdot\Delta^{-1}\big)^*$, hence we obtain
    \[
        \p_{t_p^{(a)}}L_1^M\cdot\iota_{\La^{-1}}\Delta^{-1}=\p_{t_p^{(a)}}L_2^N\cdot\iota_{\La}\Delta^{-1}+\p_{t_p^{(a)}}\Big(\sum_{l\in\Z}\sum_{i=1}^{m}q_{i,n}\La^lr_{i,n+1}\Big),
    \]
    which means the definition is well-defined.
\end{proof}

Notice that the Lax equations \eqref{Laxeq} are exactly the mToda Lax equations. Thus, as shown in~\cite{RuiCheng2024}, there exist wave operators $S_1(n,\bft,\La)=\sum_{i\geq1}s_{1,j}(n,\bft)\La^{-j}$ and $S_2(n,\bft,\La)=\sum_{i\geq1}s_{2,j}(n,\bft)\La^{j}$ satisfying
\begin{alignat}{2}
    &L_1=S_1 \cdot \La\cdot S_1^{-1} ,\qquad &&L_2 =S_2 \cdot \La^{-1}\cdot S_2^{-1},\label{eq:L1L2}\\
    &\p_{t_p^{(1)}}S_1=B_p^{(1)}\cdot S_1-S_1\cdot\La^p,\qquad &&\p_{t_p^{(1)}}S_2=B_p^{(1)}\cdot S_2,\\
    &\p_{t_p^{(2)}}S_2=B_p^{(2)}\cdot S_2-S_2\cdot \La^{-p},\qquad &&\p_{t_p^{(2)}}S_1=B_p^{(2)}\cdot S_1.\label{eq:pS2S1}
\end{alignat}
Next if we define the mToda wave functions $\Psi_a$ and the Toda adjoint wave functions $\Psi_a^*$,
\begin{equation}\label{Phi}
    \begin{alignedat}{2}
        &\Psi_1(n,\bft,z)=e^{\xi(\bft^{(1)},z)}S_1(n,\bft,\La)(z^n),\qquad &&\Psi_1^*(n,\bft,z)=-e^{-\xi(\bft^{(1)},z)}\iota_{\La^{-1}}\Delta^{-1}(S_1^*)^{-1}(z^{-n}),\\
        &\Psi_2(n,\bft,\La)=e^{\xi(\bft^{(2)},z^{-1})}S_2(n,\bft,\La)(z^n),\qquad &&\Psi_2^*(n,\bft,z)=e^{-\xi(\bft^{(2)},z^{-1})}\iota_{\La}\Delta^{-1}(S_2^*)^{-1}(z^{-n}),
    \end{alignedat}
\end{equation}
then from \cite[Theorem 3.6]{RuiCheng2024}, we have 
\be\label{mToda'}
    \oint_{C_\infty}\frac{dz}{2\pi\textnormal{\bf i} z}\Psi_1(n,\bft,z)\Psi_1^*(n',\bft',z)
    +\oint_{C_0}\frac{dz}{2\pi\textnormal{\bf i} z}\Psi_2(n,\bft,z)\Psi_2^*(n',\bft',z)=1.
\ee
Further, the following auxiliary equations also hold by \eqref{eq:L1L2}--\eqref{Phi},
\begin{alignat}{2}  
    &L_1(\Psi_1)=z\Psi_1,\qquad &&\big(\iota_{\La}\Delta^{-1}L_1^*\Delta\big)(\Psi_1^*)=z\Psi_1^*,\label{LPhi1}\\
    &L_2(\Psi_2)=z^{-1}\Psi_2,\qquad &&\big(\iota_{\La^{-1}}\Delta^{-1}L_2^*\Delta\big)(\Psi_2^*)=z^{-1}\Psi_2^*,\label{LPhi2}\\
    &\p_{t_p^{(a)}}(\Psi_b)=B_p^{(a)}(\Psi_b),\qquad &&\p_{t_p^{(a)}}(\Psi_b^*)=-\big(\iota_{\La}\Delta^{-1}B_p^{(a)*}\Delta\big)(\Psi_b^*).\label{eq:ParPsi}
\end{alignat}
Based on the above preliminaries, we have the following theorem.
\begin{theorem}\label{thm:Wave}
    The GBMT hierarchy \eqref{ConsL'}--\eqref{partialqr} also requires that the mToda wave functions $\Psi_a(n,\bft,z)$ and $\Psi_a^*(n,\bft,z)$ satisfy the following bilinear equations.
    \begin{align}
        &\begin{aligned}\label{WavePhi1}
            \oint_{C_\infty}&\frac{dz}{2\pi\text{i} z}\iota_{\La}\Delta^{-1}\Big(\Delta\big(q_{i,n}(\bft)\big)\cdot \Psi_1^*(n+1,\bft,z)\Big)\cdot\Psi_1(n',\bft',z)\\
            &+\oint_{C_0}\frac{dz}{2\pi\text{i} z}\iota_{\La^{-1}}\Delta^{-1}\Big(\Delta\big(q_{i,n}(\bft)\big)\cdot\Psi_2^*(n+1,\bft,z)\Big)\cdot\Psi_2(n',\bft',z)=q_{i,n}(\bft)-q_{i,n'}(\bft'),
        \end{aligned}\\
        &\begin{aligned}\label{WavePhi2}
            \oint_{C_\infty}&\frac{dz}{2\pi\text{i} z}\iota_{\La^{-1}}\Delta^{-1}\Big(\Delta\big(r_{i,n}(\bft)\big)\cdot\Psi_1(n,\bft,z)\Big)\cdot \Psi_1^*(n',\bft',z)\\
            &+\oint_{C_0}\frac{dz}{2\pi\text{i} z}\iota_{\La}\Delta^{-1}\Big(\Delta\big(r_{i,n}(\bft)\big)\cdot\Psi_2(n,\bft,z)\Big)\cdot\Psi_2^*(n',\bft',z)=r_{i,n}(\bft)-r_{i,n'}(\bft'),
        \end{aligned}\\
        &\begin{aligned}\label{WavePhi3}
            \oint_{C_\infty}\frac{dz}{2\pi\textnormal{\bf i} z}&z^M\Psi_1(n,\bft,z)\Psi_1^*(n',\bft',z)\\
            &+\oint_{C_0}\frac{dz}{2\pi\textnormal{\bf i} z}z^{-N}\Psi_2(n,\bft,z)\Psi_2^*(n',\bft',z)=\sum_{i=1}^{m}q_{i,n}(\bft)r_{i,n'}(\bft').
        \end{aligned}
    \end{align}
\end{theorem}
\begin{proof}
    Firstly by using \eqref{mToda'}, we know that \eqref{WavePhi1} is equivalent to 
    \begin{equation}\label{eq:q}
        \begin{aligned}
            \oint_{C_\infty}\frac{dz}{2\pi\textnormal{\bf i} z}&\Psi_1(n,\bft,z)\cdot \iota_{\La}\Delta^{-1}\Big( q_{i,n'}(\bft')\cdot \Delta\big(\Psi_1^*(n',\bft',z)\big)\Big)\\
            &+\oint_{C_0}\frac{dz}{2\pi\textnormal{\bf i} z}\Psi_2(n,\bft,z)\cdot\iota_{\La^{-1}}\Delta^{-1}\Big(q_{i,n'}(\bft')\cdot \Delta\big(\Psi_2^*(n',\bft',z)\big)\Big)=q_{i,n}(\bft).
        \end{aligned}
    \end{equation}
    Since $\p_{t_p^{(a)}}(\Psi_b)=B_p^{(a)}(\Psi_b)$ and $\p_{t_p^{(a)}}(q_{i,n})=B_p^{(a)}(q_{i,n})$, recalling Taylor expansion of \eqref{eq:q} at $\bft=\bft'$, the proof of \eqref{eq:q} is equivalent to 
    \begin{align*}
        \sum_{l\in\Z}\oint_{C_\infty}\frac{dz}{2\pi\textnormal{\bf i} z}&\Psi_1(n,\bft,z)\cdot \iota_{\La}\Delta^{-1}\Big( q_{i,n+l}(\bft)\cdot\Delta\big(\Psi_1^*(n+l,\bft,z)\big)\Big)\La^l\\
        &+\sum_{l\in\Z}\oint_{C_0}\frac{dz}{2\pi\textnormal{\bf i} z}\Psi_2(n,\bft,z)\cdot\iota_{\La^{-1}}\Delta^{-1}\Big(q_{i,n+l}(\bft)\cdot\Delta\big(\Psi_2^*(n+l,\bft,z)\big)\Big)\La^l=\sum_{l\in\Z}q_{i,n}(\bft)\La^l.
    \end{align*}
    It can be verified that the above equation holds by taking definition \eqref{Phi} into the left hand side of the above equation, and using Lemma \ref{LemForD}.  The proof of \eqref{WavePhi2} is analogous. 

    Next let us prove \eqref{WavePhi3}. Taking the adjoint operation $*$ to the constraint \eqref{ConsL'}, we have
    \begin{equation}\label{ConjConstL}
        \iota_{\La^{-1}}\Delta^{-1}(L_2^*)^N=\iota_{\La}\Delta^{-1}(L_1^*)^M+\sum_{l\in\Z}\sum_{i=1}^{m}r_{i,n}\La^lq_{i,n}.
    \end{equation}
    Applying $\Delta$ to \eqref{mToda'} with respect to $n'$, one obtains
    \begin{equation}\label{eq:DeltaPsi*}
        \oint_{C_\infty}\frac{dz}{2\pi\textnormal{\bf i} z}\Psi_1(n,\bft,z)\Delta\big(\Psi_1^*(n',\bft',z)\big)+\oint_{C_0}\frac{dz}{2\pi\textnormal{\bf i} z}\Psi_2(n,\bft,z)\Delta\big(\Psi_2^*(n',\bft',z)\big)=0.
    \end{equation}
    Then by acting the operator \eqref{ConjConstL} on \eqref{eq:DeltaPsi*} with respect to $n'$, and using \eqref{LPhi1} and \eqref{LPhi2}, we have
    \begin{equation}\label{eq:Psiqr}
        \begin{aligned}
            &\oint_{C_\infty}\frac{dz}{2\pi\textnormal{\bf i} z}z^M\Psi_1(n,\bft,z)\Psi_1^*(n',\bft',z)+\oint_{C_0}\frac{dz}{2\pi\textnormal{\bf i} z}z^{-N}\Psi_2(n,\bft,z)\Psi_2^*(n',\bft',z)\\
            =&-\sum_{l\in\Z}\sum_{i=1}^{m}\oint_{C_\infty}\frac{dz}{2\pi\textnormal{\bf i} z}\Psi_1(n,\bft,z)\cdot r_{i,n'}(\bft')\cdot\La^l\Big(q_{i,n'}(\bft')\cdot\Delta\big(\Psi_1^*(n',\bft',z)\big)\Big).
        \end{aligned}
    \end{equation}
    To match the desired equation \eqref{WavePhi3}, it is sufficient to prove
    \begin{equation}\label{Phiq}
        \begin{aligned}
            -\sum_{l\in\Z}\oint_{C_\infty}\frac{dz}{2\pi\textnormal{\bf i} z}\Psi_1(n,\bft,z)\cdot\La^l\Big( q_{i,n'}(\bft')\cdot\Delta\big(\Psi_1^*(n',\bft',z)\big)\Big)=q_{i,n}(\bft).
        \end{aligned}
    \end{equation}
    Indeed, using $\iota_{\La}\Delta^{-1}=-\sum_{k\geq0}\La^k$, $\iota_{\La^{-1}}\Delta^{-1}=\sum_{k\geq1}\La^{-k}$ and \eqref{eq:DeltaPsi*}, the left side of \eqref{Phiq} can be written as
    \begin{align*}
        &\oint_{C_\infty}\frac{dz}{2\pi\textnormal{\bf i} z}\Psi_1(n,\bft,z)\cdot(\iota_{\La^{-1}}\Delta^{-1}-\iota_{\La}\Delta^{-1})\Big(q_{i,n'}(\bft')\cdot \Delta\big(\Psi_1^*(n',\bft',z)\big)\Big)\\
        &=-\oint_{C_\infty}\frac{dz}{2\pi\textnormal{\bf i} z}\Psi_1(n,\bft,z)\cdot\iota_{\La}\Delta^{-1}\Big(q_{i,n'}(\bft')\cdot \Delta\big(\Psi_1^*(n',\bft',z)\big)\Big)\\
        &\quad+\oint_{C_0}\frac{dz}{2\pi\textnormal{\bf i} z}\Psi_2(n,\bft,z)\cdot\iota_{\La^{-1}}\Delta^{-1}\Big(q_{i,n'}(\bft')\cdot \Delta\big(\Psi_2^*(n',\bft',z)\big)\Big)=q_{i,n}(\bft),
    \end{align*}
    where the last the identity holds because of \eqref{eq:q}. Hence equation \eqref{eq:Psiqr} holds.
\end{proof}

\begin{remark}\label{Rmk:spectral}
    Here \eqref{WavePhi1} or \eqref{eq:q} is the spectral representation of the mToda eigenfunction $q_{i,n}(\bft)$, satisfying $\p_{t_p^{(a)}}(q_{i,n})=B_p^{(a)}(q_{i,n})$. While \eqref{Wave1'} or \eqref{WavePhi2} is the spectral representation of the mToda adjoint eigenfunction $r_{i,n}(\bft)$, satisfying $\p_{t_p^{(a)}}(r_{i,n})=-\big(\Delta^{-1} B_{p}^{(a)*}\Delta\big)(r_{i,n})$. We would remark that the spectral representation of the mToda (adjoint) eigenfunctions can be directly derived from the mToda hierarchy without considering the GBMT constraint \eqref{ConsL'}.
\end{remark}

If we apply $\Delta$ to \eqref{WavePhi1}, we can obtain the mToda bilinear equation in terms of wave functions, that is \eqref{mToda'}. Thus, as shown in \cite{RuiCheng2024}, there exists the mToda tau pair $(\tau_{0,n},\tau_{1,n})$ such that
\begin{equation}\label{eq:PsiPsi*}
    \begin{aligned}
        &\Psi_1(n,\bft,z)=\frac{\tau_{0,n}(\bft-[z^{-1}]_1)}{\tau_{1,n}(\bft)}z^ne^{\xi(\bft^{(1)},z)},\qquad \Psi_1^*(n,\bft,z)=\frac{\tau_{1,n}(\bft+[z^{-1}]_1)}{\tau_{0,n}(\bft)}z^{-n}e^{-\xi(\bft^{(1)},z)},\\
        &\Psi_2(n,\bft,z)=\frac{\tau_{0,n+1}(\bft-[z]_2)}{\tau_{1,n}(\bft)}z^ne^{\xi(\bft^{(2)},z^{-1})},\qquad \Psi_2^*(n,\bft,z)=\frac{\tau_{1,n-1}(\bft+[z]_2)}{\tau_{0,n}(\bft)}z^{1-n}e^{-\xi(\bft^{(2)},z^{-1})}.
    \end{aligned}
\end{equation}

\begin{lemma}\label{Lem:eq}
    For the mToda wave functions $\Psi_a$, the mToda adjoint wave functions $\Psi^*_a$, the mToda eigenfunctions $q_{i,n}$ and the mToda adjoint eigenfunctions $r_{i,n}$, the following identities hold:
    \begin{align*}
        &\iota_{\La}\Delta^{-1}\Big(\Delta\big(q_{i,n}(\bft)\big)\Psi_1^*(n+1,\bft,z)\Big)=\Big(q_{i,n}(\bft)-q_{i,n}(\bft+[z^{-1}]_1)\Big)\Psi_1^*(n,\bft,z).\\
        &\iota_{\La^{-1}}\Delta^{-1}\Big(\Delta\big(q_{i,n}(\bft)\big)\Psi_2^*(n+1,\bft,z)\Big)=\Big(q_{i,n}(\bft)-q_{i,n-1}(\bft+[z]_2)\Big)\Psi_2^*(n,\bft,z),\\
        &\iota_{\La^{-1}}\Delta^{-1}\Big(\Delta\big(r_{i,n+1}(\bft)\big)\Psi_1(n+1,\bft,z)\Big)=\Big(r_{i,n+1}(\bft)-r_{i,n}(\bft-[z^{-1}]_1)\Big)\Psi_1(n,\bft,z),\\
        &\iota_{\La}\Delta^{-1}\Big(\Delta\big(r_{i,n+1}(\bft)\big)\Psi_2(n+1,\bft,z)\Big)=\Big(r_{i,n+1}(\bft)-r_{i,n+1}(\bft-[z]_2)\Big)\Psi_2(n,\bft,z).
    \end{align*}
\end{lemma}
\begin{proof}
    Firstly, by using \eqref{eq:PsiPsi*}, one can let 
    \begin{alignat*}{2}
        &\Psi_1(n,\bft,z)=\tilde{\psi}_1(n,\bft,z)z^ne^{\xi(\bft^{(1)},z)},\qquad &&\Psi_1^*(n,\bft,z)=\tilde{\psi}_1^*(n,\bft,z)z^{-n}e^{-\xi(\bft^{(1)},z)},\\
        &\Psi_2(n,\bft,z)=\tilde{\psi}_2(n,\bft,z)z^ne^{\xi(\bft^{(2)},z^{-1})},\qquad &&\Psi_2^*(n,\bft,z)=\tilde{\psi}_2^*(n,\bft,z)z^{-n}e^{-\xi(\bft^{(2)},z^{-1})},
    \end{alignat*}
    where $\tilde{\psi}_1=\sum_{k\geq0}\tilde{\psi}_{1,k}z^{-k},~\tilde{\psi}_1^*=\sum_{k\geq0}\tilde{\psi}_{1,k}^*z^{-k},~\tilde{\psi}_2=\sum_{k\geq0}\tilde{\psi}_{2,k}z^k$ and $\tilde{\psi}_2^*=\sum_{k\geq0}\tilde{\psi}_{2,k}^*z^{k+1}.$ Further, from \eqref{eq:PsiPsi*} we obtain
    \begin{equation}\label{eq:psi1psi1*}
        \tilde{\psi}_1(n,\bft+[z^{-1}]_1,z)=\tilde{\psi}_1^*(n,\bft,z)^{-1}.
    \end{equation}
    Thus we can assume 
    \begin{align*}
        &\iota_{\La}\Delta^{-1}\Big(\Delta\big(q_{i,n}(\bft)\big)\Psi_1^*(n+1,\bft,z)\Big)=z^{-n}e^{-\xi(\bft^{(1)},z)}K_1(n,\bft,z)=\sum_{j\geq1}z^{-n}e^{-\xi(\bft^{(1)},z)}K_1(n,\bft)z^{-j},\\
        &\iota_{\La^{-1}}\Delta^{-1}\Big(\Delta\big(q_{i,n}(\bft)\big)\Psi_2^*(n+1,\bft,z)\Big)=z^{-n}e^{-\xi(\bft^{(2)},z^{-1})}K_2(n,\bft,z)=\sum_{j\geq1}z^{-n}e^{-\xi(\bft^{(2)},z^{-1})}K_2(n,\bft)z^{j}.
    \end{align*}

    The following equation is to be used next,
    \begin{equation}\label{eq:Residue'}
        \Res_z\sum_{k=0}^{\infty}a_k(\la)z^{-k}\frac{1}{1-z/\la}=\la\sum_{k=1}^{\infty}a_k(\la)
    z^{-k}\Big|_{z=\la}.
    \end{equation}
    Notice that the mToda eigenfunction $q_{i,n}$ in \eqref{partialqr} satisfies the spectral representation \eqref{WavePhi1} (please see Remark \ref{Rmk:spectral}). Thus letting $n'=n$ and $\bft'=\bft+[\la^{-1}]_1$ for \eqref{WavePhi1}, note that 
    \begin{align*}
        &\oint_{C_0}\frac{dz}{2\pi\text{i} z}\iota_{\La^{-1}}\Delta^{-1}\Big(\Delta\big(q_{i,n}(\bft)\big)\cdot\Psi_2^*(n+1,\bft,z)\Big)\cdot\Psi_2(n,\bft+[\la^{-1}]_1,z)\\
        =&\oint_{C_0}\frac{dz}{2\pi\text{i} z}K_2(n,\bft,z)\tilde{\psi}_2^*(n,\bft+[\la^{-1}]_1,z)=0,
    \end{align*}
    where the second identity holds due to \eqref{eq:Residue'}. Thus by using \eqref{WavePhi1},
    \begin{align*}
        q_{i,n}(\bft)-q_{i,n}(\bft+[\la^{-1}]_1)=&\oint_{C_\infty}\frac{dz}{2\pi\text{i} z}\iota_{\La}\Delta^{-1}\Big(\Delta\big(q_{i,n}(\bft)\big)\cdot \Psi_1^*(n+1,\bft,z)\Big)\cdot\Psi_1(n,\bft+[\la^{-1}]_1,z)\\
        =&\oint_{C_\infty}\frac{dz}{2\pi\text{i} z}K_1(n,\bft,z)\tilde{\psi}_1(n,\bft+[\la^{-1}]_1,z)\frac{1}{1-z/\la}\\
        =&K_1(n,\bft,\la)\tilde{\psi}_1(n,\bft+[\la^{-1}]_1,\la)\\
        =&\la^{-n}e^{-\xi(\bft^{(1)},\la)}K_1(n,\bft,\la)\cdot\la^{n}e^{\xi(\bft^{(1)},\la)}\tilde{\psi}_1^*(n,\bft,\la)^{-1},
    \end{align*}
    where the third identity holds due to \eqref{eq:Residue'} and the fourth identity holds due to \eqref{eq:psi1psi1*}. Hence, we have
    \[
        \iota_{\La}\Delta^{-1}\Big(\Delta\big(q_{i,n}(\bft)\big)\Psi_1^*(n+1,\bft,\la)\Big)=(q_{i,n}(\bft)-q_{i,n}(\bft+[\la^{-1}]_1))\Psi_1^*(n,\bft,\la).
    \]
    The remaining identities follow by analogous arguments.
\end{proof}

If we denote 
\begin{equation}\label{eq:rhosigma}
    \rho_{i,n}(\bft)=q_{i,n}(\bft)\tau_{1,n}(\bft),\qquad \sigma_{i,n}(\bft)=r_{i,n}(\bft)\tau_{0,n}(\bft),
\end{equation}
then by Lemma \ref{Lem:eq} and Theorem \ref{thm:Wave}, we can obtain the bilinear equations for the GBMT hierarchy in terms of the mToda tau pair, which is given by the following corollary.
\begin{corollary}
    The functions $\tau_{0,n},\tau_{1,n}$ defined by \eqref{eq:PsiPsi*} and $\rho_{i,n}, \sigma_{i,n}$ defined by \eqref{eq:rhosigma} satisfy \eqref{cmToda1}--\eqref{mToda}.
\end{corollary}

\begin{remark}
    We have established the following equivalent formulations of the GBMT hierarchy, 
    \begin{itemize}
        \item bilinear equations in terms of tau functions \eqref{cmToda1}--\eqref{mToda};
        \item bilinear equations in terms of wave functions \eqref{WavePhi1}--\eqref{WavePhi3};
        \item Lax formulations \eqref{ConsL'}--\eqref{partialqr}.
    \end{itemize}
\end{remark}

\section{Miura links with the GBT hierarchy}\label{Sect:Miura}

Note that the mToda hierarchy is related to the Toda hierarchy through Miura links\cite{RuiCheng2024}. In this section, we discuss the relations between the GBMT hierarchy and the GBT hierarchy.

Recall that there are two Miura links between the Toda and mToda hierarchies. Given Lax operators $(L_1,L_2)$ of the mToda hierarchy, let
\begin{equation}\label{eq:Miura'}
    \begin{alignedat}{2}
        &\LL_1=T_1\cdot L_1\cdot T_1^{-1}, \qquad &&\LL_2=T_1\cdot L_2\cdot T_1^{-1}, \\
        &\LL_1'=T_2\cdot L_1\cdot \iota_{\La^{-1}}T_2^{-1}, \qquad &&\LL_2'=T_2\cdot L_2\cdot \iota_{\La}T_2^{-1},
    \end{alignedat}
\end{equation}
where 
\begin{equation}
    T_1=c_0^{-1}(n),\qquad T_2=c_0^{-1}(n+1)\Delta, \qquad c_0(n)=(S_1)_{[0]}.
\end{equation}
Here $S_1$ is the wave operator of $L_1$, that is $L_1=S_1\La S_1^{-1}$. Then \cite[Section 4]{RuiCheng2024} shows us the following results:
\begin{itemize}
    \item $(\LL_1,\LL_2)$ and $(\LL_1',\LL_2')$ are two different Lax operators of the Toda hierarchy, i.e.
    \begin{equation}
        \p_{t^{(a)}_p}\tilde{\LL}_b=[\tilde{\mathcal{B}}_p^{(a)},\tilde{\LL}_b], \qquad \tilde{\LL}_a=\LL_a~\textnormal{or}~\LL_a',\quad a,b=1,2,\label{eq:TodaLax}
    \end{equation}
    where $\tilde{\mathcal{B}}_p^{(1)}=(\tilde{\LL}_1^p)_{\geq0},~ \tilde{\mathcal{B}}_p^{(2)}=(\tilde{\LL}_2^p)_{<0}$ and $\tilde{\LL}_1=\Lambda +\sum_{i\geq0}\tilde{a}_i\La^{-i}$ and $\tilde{\LL}_2=\tilde{b}\Lambda^{-1}+\sum_{i\geq0}\tilde{b}_i\La^{i}$.
    \item $c^{-1}_0(n)$ is the Toda eigenfunction of $(\LL_1,\LL_2)$, and $c_0(n+1)$ is the Toda adjoint eigenfunction of $(\LL_1',\LL_2')$, that is,
    \[
        \p_{t^{(a)}_p}c^{-1}_0(n)=\mathcal{B}_p^{(a)}(c^{-1}_0(n)),\qquad \p_{t^{(a)}_p}c_0(n+1)=-(\mathcal{B}_p^{'(a)})^*(c_0(n+1)).
    \]
\end{itemize}
Here we use prime to stress that the two transformations give different results. The other symbols are expressed similarly.

We can deduce that the constraint of the GBT hierarchy from the constraint of the GBMT hierarchy. More precisely, we need following preparation.
\begin{lemma}\label{lem:B&BB}
    Under above transformations, we get
    \begin{align}
        &B_p^{(a)}=c_0\cdot \mathcal{B}_p^{(a)}\cdot c_0^{-1}+\p_{t_p^{(a)}}(c_0)\cdot c_0^{-1},\label{eq:B&BB}\\
        &B_p'^{(a)}=\iota_{\La^{-1}}\Delta^{-1}\cdot\Big(c_0(n+1)\cdot\mathcal{B}_p'^{(a)}+\p_{t_p^{(a)}}(c_0(n+1))\Big)\cdot c_0^{-1}(n+1) \Delta,\label{eq:B&BB'}
    \end{align}
    where $B_p^{(1)}=(L_1^p)_{\Delta,\geq1}$, $B_p^{(2)}=(L_2^p)_{\Delta^*,\geq1}$.
    \begin{proof}
        While \eqref{eq:B&BB} is proved in \cite[Proposition 4.6]{RuiCheng2024}, we provide here a proof for $B_p'^{(1)}$. Recalling definition \eqref{eq:Miura'} and 
        \begin{equation}
            \Big(\iota_{\La^{-1}}\Delta^{-1}fA f^{-1} \Delta\Big)_{\Delta,\geq1}=\iota_{\La^{-1}}\Delta^{-1}fA_{\geq0}f^{-1}\Delta-\iota_{\La^{-1}}\Delta^{-1}(A_{\geq0})^*(f)f^{-1}\Delta,\label{eq:fAf}
        \end{equation}
        we get 
        \begin{align*}
            B_p'^{(1)}=&\Big(\iota_{\La^{-1}}\Delta^{-1}c_0(n+1)\LL_1'^p c_0^{-1}(n+1) \Delta\Big)_{\Delta,\geq1}\\
            =&\iota_{\La^{-1}}\Delta^{-1}c_0(n+1)\mathcal{B}_p'^{(1)}c_0^{-1}(n+1)\Delta+\iota_{\La^{-1}}\Delta^{-1}\p_{t_p^{(1)}}(c_0(n+1))c_0^{-1}(n+1)\Delta.
        \end{align*}
        So we get the formula, and the method is also valid for $a=2$.
    \end{proof}
\end{lemma}
\begin{lemma}\label{lem:eigenfunc}
    Let $q_n$ be mToda eigenfunction and $r_n$ be mToda adjoint eigenfunction, then
    \begin{itemize}
        \item $c_0^{-1}\cdot q_n$ is Toda eigenfunction and $c_0\cdot \Delta(r_n)$ is Toda adjoint eigenfunction with respect to $(\LL_1,\LL_2)$;
        \item $c_0^{-1}(n+1)\cdot \Delta(q_n)$ is Toda eigenfunction and $c_{0}(n+1)\cdot r_{n+1}$ is Toda adjoint eigenfunction with respect to $(\LL_1',\LL_2')$.
    \end{itemize}
\end{lemma}
\begin{proof}
    This proof is straightforward with Lemma \ref{lem:B&BB}.
\end{proof}

With the above preparations, we now describe the Miura links between the two constraints.
\begin{theorem}
    Let $(L_1,L_2)$ be the Lax operators of GBMT hierarchy, i.e.,
    \begin{equation}
        L_1^M\iota_{\La^{-1}}\Delta^{-1}=L_2^N\iota_{\La}\Delta^{-1}+\sum_{l\in\Z}\sum_{i=1}^{m}q_{i,n}\La^lr_{i,n+1}.\label{eq:consL'}
    \end{equation}
    Then operators $(\LL_1,\LL_2)$ and $(\LL_1',\LL_2')$ defined by \eqref{eq:Miura'} are the GBT Lax operators,
    \begin{equation}
        \LL_1^M=\LL_2^{N}+\sum_{l\in \mathbb Z}\sum_{i=1}^{m}\tilde{q}_{i,n}\Lambda^l\tilde{r}_{i,n}, \qquad \LL_1'^M=\LL_2'^N+\sum_{l\in\Z}\sum_{i=1}^{m}\tilde{q}_{i,n}'\Lambda^l\tilde{r}_{i,n+1}'.\label{eq:consLL'}
    \end{equation}
\end{theorem}
\begin{proof}
    Recalling the first equation of \eqref{eq:ConsL1}, we get
    \begin{align*}
        L_1^M-L_2^N=&\sum_{l\in\Z}\sum_{i=1}^{m}q_{i,n}\big(\La^{l+1}\cdot r_{i,n}-\La^{l}\cdot r_{i,n+1}\big)\\
        =&\sum_{l\in\Z}\sum_{i=1}^{m}q_{i,n}\La^{l}\cdot\big(r_{i,n}- r_{i,n+1}\big)
        =-\sum_{l\in\Z}\sum_{i=1}^{m}q_{i,n}\La^{l}\cdot \Delta(r_{i,n}).
    \end{align*}
    For the first transformation $T_1$,  
    \[
        c_0^{-1}\cdot L_1^M\cdot c_0=c_0^{-1}\cdot L_2^N\cdot c_0-\sum_{l\in\Z}\sum_{i=1}^{m}c_0^{-1}\cdot q_{i,n}\La^l\cdot \Delta(r_{i,n})\cdot c_0.
    \]
    Therefore
    \[
        \LL_1^M=\LL_2^{N}+\sum_{l\in \mathbb Z}\sum_{i=1}^{m}\tilde{q}_{i,n}\Lambda^l\tilde{r}_{i,n},
    \]
    where $\tilde{q}_{i,n}=c_0^{-1}\cdot q_{i,n}$ is Toda eigenfunction and $\tilde{r}_{i,n}=\Delta(r_{i,n})\cdot c_0$ is Toda adjoint eigenfunction according to Lemma \ref{lem:eigenfunc}.

    For the second transformation $T_2$, recalling Lemma \ref{Asum}, we have 
    \begin{align*}
        &c_0^{-1}(n+1)\Delta\cdot L_1^M\iota_{\La^{-1}}\Delta^{-1}\cdot c_0(n+1)\\
        &\qquad =c_0^{-1}(n+1)\Delta\cdot L_2^N\iota_{\La}\Delta^{-1}\cdot c_0(n+1)+\sum_{l\in\Z}\sum_{i=1}^{m}c_0^{-1}(n+1)\Delta(q_{i,n})\La^lr_{i,n+1}c_0(n+1).
    \end{align*}
    Thus 
    \[
        \LL_1'^M=\LL_2'^N+\sum_{l\in\Z}\sum_{i=1}^{m}\tilde{q}_{i,n}'\Lambda^l\tilde{r}_{i,n+1}',
    \]
    where $\tilde{q}_{i,n}'=c_0^{-1}(n+1)\Delta(q_{i,n})$ is Toda eigenfunction and $\tilde{r}_{i,n+1}'=c_0(n+1)r_{i,n+1}$ is Toda adjoint eigenfunction by using Lemma \ref{lem:eigenfunc}.
\end{proof}

\section{Equivalent conditions on tau functions for the GBMT hierarchy}\label{sect:ConsInTau}

In this section, we establish the equivalent constraints on tau functions for the GBT and GBMT hierarchies. Recall that the constraints on Lax operators of the GBMT and GBT hierarchies are given as follows.
\begin{itemize}
    \item The GBMT hierarchy: 
    \begin{equation*}
        L_1(n)^M=L_2(n)^N+\sum_{l\in\Z}\sum_{i=1}^{m}q_{i,n}\cdot\La^l\cdot r_{i,n+1}\cdot\Delta,\qquad \big(L_1^M(n)+L_2^N(n)\big)(1)=0.
    \end{equation*}
    \item The GBT hierarchy: 
    \begin{equation*}
        \LL_1(n)^M=\LL_2(n)^{N}+\sum_{j\in \mathbb Z}\sum_{i=1}^{m}\tilde{q}_{i,n}\cdot\Lambda^j\cdot\tilde{r}_{i,n}.
    \end{equation*} 
\end{itemize}

Firstly let us discuss the case of the GBMT hierarchy.

\begin{proposition}\label{prop:Fayqr}
    For the mToda wave functions $\Psi_a$, the mToda adjoint wave functions $\Psi^*_a$, the mToda eigenfunctions $q_{i,n}$ and the mToda adjoint eigenfunctions $r_{i,n}$ defined in \eqref{eq:PsiPsi*}--\eqref{eq:rhosigma} satisfying \eqref{WavePhi1}--\eqref{WavePhi3}, the following identities hold:
    \begin{equation}\label{I=Delta}
        \begin{aligned}
            &\oint_{C_\infty}\frac{dz}{2\pi {\bf i}z}\iota_{\La^{-1}}\Delta^{-1}\Big(r_{i,n+1}(\bft)\Delta\big(\Psi_1(n,\bft,z)\big)\Big)\cdot\iota_{\La}\Delta^{-1} \Big(q_{i,n'}(\bft')\Delta\big(\Psi_1^*(n',\bft',z)\big)\Big)\\
            &+\oint_{C_0}\frac{dz}{2\pi {\bf i}z}\iota_{\La}\Delta^{-1}\Big(r_{i,n+1}(\bft)\Delta\big(\Psi_2(n,\bft,z)\big)\Big)\cdot\iota_{\La^{-1}}\Delta^{-1}\Big(q_{i,n'}(\bft')\Delta\big(\Psi_2^*(n',\bft',z)\big)\Big)\\
            =&\Delta^{-1}\Big(\Delta(q_{i,n}(\bft)) r_{i,n+1}(\bft)\Big)+\Delta^{-1}\Big(q_{i,n'}(\bft') \Delta(r_{i,n'+1}(\bft'))\Big).
        \end{aligned}
    \end{equation}
    and 
    \begin{equation}\label{eq:Deltaqr}
        \begin{aligned}
            &\Delta\big(q_{i,n}(\bft)\big) r_{i,n+1}(\bft)+q_{i,n}(\bft-[\la^{-1}]_1) \Delta\big(r_{i,n+1}(\bft-[\la^{-1}]_1)\big)=\Delta\Big(q_{i,n}(\bft)r_{i,n}(\bft-[\la^{-1}]_1)\Big),\\
            &\Delta\big(q_{i,n}(\bft)\big) r_{i,n+1}(\bft)+q_{i,n+1}(\bft-[\la]_2) \Delta\big(r_{i,n+2}(\bft-[\la]_2)\big)=\Delta\Big(q_{i,n}(\bft)r_{i,n+1}(\bft-[\la]_2)\Big).
        \end{aligned}
    \end{equation}
\end{proposition}
\begin{proof}
    Denote the left hand side of \eqref{I=Delta} as $I(n,n',\bft,\bft')$. Recalling the spectral representations \eqref{Wave} and \eqref{Wave1'}, we have
    \begin{equation}\label{eq:DeltaI}
        \begin{aligned}
            &I(n+1,n',\bft,\bft')-I(n,n',\bft,\bft')=\Delta(q_{i,n}(\bft))\cdot r_{i,n+1}(\bft),\\
            &I(n,n'+1,\bft,\bft')-I(n,n',\bft,\bft')=q_{i,n'}(\bft')\cdot \Delta(r_{i,n'+1}(\bft')).
        \end{aligned}
    \end{equation}
    By the first equation in \eqref{eq:DeltaI}, there exists an unknown function $f(n')$, such that
    \[
        I=\Delta^{-1}\big(\Delta(q_{i,n}(\bft)) r_{i,n+1}(\bft)\big)+f(n').
    \]
    Then combining the second equation in \eqref{eq:DeltaI}, there is $\Delta f(n')=q_{i,n'}(\bft') \Delta(r_{i,n'+1}(\bft'))$, which leads to \eqref{I=Delta}. 

    Next let us prove \eqref{eq:Deltaqr}, and the following equation will be used
    \begin{equation}\label{eq:Residue}
        \Res_z\sum_{k=0}^{\infty}a_k(\la)z^{-k}\frac{1}{1-z/\la}=\la\sum_{k=1}^{\infty}a_k(\la)z^{-k}\Big|_{z=\la}.
    \end{equation}
    Let $\bft'=\bft-[\la^{-1}]_1$ and $n'=n$ for \eqref{I=Delta}. Recalling definitions \eqref{eq:PsiPsi*} and using Lemma \ref{Lem:eq}, we get
    \begin{align*}
        &\oint_{C_\infty}\frac{dz}{2\pi {\bf i}z}r_{i,n}(\bft-[z^{-1}]_1)q_{i,n}(\bft-[\la^{-1}]_1+[z^{-1}]_1)\Psi_1(n,\bft,z)\Psi_1^*(n,\bft-[\la^{-1}]_1,z)\\
        =&\oint_{C_\infty}\frac{dz}{2\pi {\bf i}z}r_{i,n}(\bft-[z^{-1}]_1)q_{i,n}(\bft-[\la^{-1}]_1+[z^{-1}]_1)\frac{\tau_{0,n}(\bft-[z^{-1}]_{1})}{\tau_{1,n}(\bft)}\frac{\tau_{1,n}(\bft-[\la^{-1}]_1+[z^{-1}]_{1})}{\tau_{0,n}(\bft-[\la^{-1}]_1)}\frac{1}{1-z/\la}\\
        =&q_{i,n}(\bft)r_{i,n}(\bft-[\la^{-1}]_1),
    \end{align*}
    where we have used \eqref{eq:Residue}, and
    \begin{align*}
        \oint_{C_0}\frac{dz}{2\pi {\bf i}}r_{i,n+1}(\bft-[z]_2)q_{i,n-1}(\bft-[\la^{-1}]_1+[z]_2)\Psi_2(n,\bft,z)\Psi_2^*(n,\bft-[\la^{-1}]_1,z)=0.
    \end{align*}
    Thus 
    \[
        \Delta^{-1}\Big(\Delta\big(q_{i,n}(\bft)\big) r_{i,n+1}(\bft)\Big)+\Delta^{-1}\Big(q_{i,n}(\bft-[\la^{-1}]_1) \Delta\big(r_{i,n+1}(\bft-[\la^{-1}]_1)\big)\Big)=q_{i,n}(\bft)r_{i,n}(\bft-[\la^{-1}]_1).
    \]
    The second identity can be proved in a similar way if $\bft'=\bft-[\la]_2,\ n'=n+1$ is considered in \eqref{I=Delta}.
\end{proof}

\begin{theorem}\label{Thm:DDelta}
    Let $D_{M,N}=\p_{t_M^{(1)}}+\p_{t_N^{(2)}}$, the GBMT constraint 
    \begin{equation}
        L_1(n)^M=L_2(n)^N+\sum_{l\in\Z}\sum_{i=1}^{m}q_{i,n}\cdot\La^l\cdot r_{i,n+1}\cdot\Delta,\qquad \big(L_1^M(n)+L_2^N(n)\big)(1)=0.\label{eq:consL1'}
    \end{equation}
    is equivalent to 
    \begin{align}
        D_{M,N}\Delta\log\tau_{0,n}=-\sum_{i=1}^{m}q_{i,n}\Delta(r_{i,n}),\qquad D_{M,N}\Delta\log\tau_{1,n}=\sum_{i=1}^{m}\Delta(q_{i,n})r_{i,n+1},\label{eq:Dlogtau}
    \end{align}
    where $(\tau_{0,n},\tau_{1,n})$ is tau pair of the mToda hierarchy.
\end{theorem}
\begin{proof}
    Firstly let us prove \eqref{eq:consL1'}$\Rightarrow$\eqref{eq:Dlogtau}. Recalling Corollary \ref{coro:conL'}, the GBMT constraint \eqref{eq:consL1'} is equivalent to 
    \begin{equation}\label{eq:consL1L2}
        \begin{aligned}
            &L_1(n)^M=B_M^{(1)}(n)+B_N^{(2)}(n)+\sum_{i=1}^{m}q_{i,n}\cdot\iota_{\La^{-1}}\Delta^{-1}\cdot r_{i,n+1}\cdot\Delta,\\
            &L_2(n)^N=B_M^{(1)}(n)+B_N^{(2)}(n)+\sum_{i=1}^{m}q_{i,n}\cdot\iota_{\La}\Delta^{-1}\cdot r_{i,n+1}\cdot\Delta.
        \end{aligned}
    \end{equation}
    Then by using \eqref{LPhi1}--\eqref{eq:ParPsi} and Lemma \ref{Lem:eq}, we get 
    \[
        D_{M,N}(\Psi_1)=\Big(z^M-\sum_{i=1}^{m}q_{i,n}(\bft)r_{i,n}(\bft-[z^{-1}]_1)\Big)\cdot\Psi_1,
    \]
    Further using \eqref{Psi1}, there is
    \[
        D_{M,N}\log\tau_{0,n}(\bft-[z^{-1}]_1)-D_{M,N}\log\tau_{1,n}(\bft)=-\sum_{i=1}^{m}q_{i,n}(\bft)r_{i,n}(\bft-[z^{-1}]_1).
    \]
    Thus, from the coefficients of $z^0$, we obtain
    \begin{equation}
        D_{M,N}\log\tau_{0,n}(\bft)-D_{M,N}\log\tau_{1,n}(\bft)=-\sum_{i=1}^{m}q_{i,n}(\bft)r_{i,n}(\bft).\label{eq:Dlogn}
    \end{equation}
    For the same method for $\Psi_2$, there is 
    \begin{equation}
        D_{M,N}\log\tau_{0,n+1}(\bft)-D_{M,N}\log\tau_{1,n}(\bft)=-\sum_{i=1}^{m}q_{i,n}(\bft)r_{i,n+1}(\bft).\label{eq:Dlogn+1}
    \end{equation}
    Combining \eqref{eq:Dlogn} and \eqref{eq:Dlogn+1}, it is straightforward to get the first identity in \eqref{eq:Dlogtau}, and 
    \[
        D_{M,N}\Delta\log\tau_{1,n}(\bft)=\sum_{i=1}^{m}\Big(\Delta\big(q_{i,n}(\bft)r_{i,n}(\bft)\big)-q_{i,n}\Delta(r_{i,n})\Big).
    \]
    Notice that $\Delta(q_{i,n}r_{i,n})-q_{i,n}\Delta(r_{i,n})=\Delta(q_{i,n})r_{i,n+1}$, so we get the second identity in \eqref{eq:Dlogtau}.

    Conversely, using \eqref{eq:Dlogtau} and Proposition \ref{prop:Fayqr}, we get
    \begin{align*}
        &D_{M,N}\log\tau_{0,n}(\bft-[z^{-1}]_1)-D_{M,N}\log\tau_{1,n}(\bft)=-\sum_{i=1}^{m}q_{i,n}(\bft)r_{i,n}(\bft-[z^{-1}]_1),\\
        &D_{M,N}\log\tau_{0,n+1}(\bft-[z]_2)-D_{M,N}\log\tau_{1,n}(\bft)=-\sum_{i=1}^{m}q_{i,n}(\bft)r_{i,n+1}(\bft-[z]_2).
    \end{align*}
    Then, recalling \eqref{LPhi1}--\eqref{eq:ParPsi}, \eqref{eq:PsiPsi*} and Lemma \ref{Lem:eq}, we obtain
    \begin{align*}
        \Big((L_1^M)_{\Delta,\geq1}+(L_2^N)_{\Delta^*,\geq1}\Big)(\Psi_1)=\Big(L_1^M-\sum_{i=1}^{m}q_{i,n}(\bft)\iota_{\La^{-1}}\Delta^{-1}r_{i,n+1}(\bft)\Delta\Big)(\Psi_1),\\
        \Big((L_1^M)_{\Delta,\geq1}+(L_2^N)_{\Delta^*,\geq1}\Big)(\Psi_2)=\Big(L_2^M-\sum_{i=1}^{m}q_{i,n}(\bft)\iota_{\La}\Delta^{-1}r_{i,n+1}(\bft)\Delta\Big)(\Psi_2).
    \end{align*}
    Thus \eqref{eq:consL1L2} holds, where we have used the fact that if $A(\Psi_1)=0$ (or $B(\Psi_2)=0$) for $A\in\mathcal{A}[[\La^{-1}]][\La]$ (or $B\in\mathcal{A}[[\La]][\La^{-1}]$), then $A=0$ (or $B=0$). It means that the GBMT constraint \eqref{eq:consL1'} holds by Corollary \ref{coro:conL'}.
\end{proof}

\begin{theorem}
    The GBT constraint 
    \begin{equation}
        \LL_1(n)^M=\LL_2(n)^{N}+\sum_{j\in \mathbb Z}\sum_{i=1}^{m}\tilde{q}_{i,n}\cdot\Lambda^j\cdot\tilde{r}_{i,n}\label{eq:consLL''}
    \end{equation} 
    is equivalent to 
    \begin{equation}\label{eq:DDeltatau}
        D_{M,N}\Delta\big(\log\tau_n(\bft)\big)=\sum_{i=1}^{m}\tilde{q}_{i,n}(\bft)\tilde{r}_{i,n}(\bft).
    \end{equation}
    Here $\tau_n(\bft)$ is the Toda tau function.
\end{theorem}
\begin{proof}
    The proof is similar to the proof of Theorem \ref{Thm:DDelta}. Here we only give a few key formulas in the proof, that is,
    \begin{align*}
        &\begin{aligned}
            q_{i,n}(\bft)=&-\oint_{C_\infty}\frac{dz}{2\pi {\bf i}z}\Phi_1(n,\bft,z)\cdot\iota_{\La}\Delta^{-1}\Big(\Phi^*_1(n',\mathbf{t}',z)\cdot q_{i,n'}(\mathbf{t}')\Big)\\
            &+\oint_{C_0}\frac{dz}{2\pi {\bf i}z}\Phi_2(n,\bft,z)\cdot\iota_{\La^{-1}}\Delta^{-1}\Big(\Phi_2^*(n',\mathbf{t}',z)\cdot q_{i,n'}(\mathbf{t}')\Big),
        \end{aligned}\\
        &\begin{aligned}
            r_{i,n}(\bft)=&\oint _{C_\infty}\frac{dz}{2\pi {\bf i}z}\Phi_1^*(n,\bft,z)\cdot\iota_{\La^{-1}}\Delta^{-1}\Big(\Phi_1(n',\mathbf{t}',z)\cdot r_{i,n'}(\mathbf{t}')\Big)\\
            &-\oint_{C_0}\frac{dz}{2\pi {\bf i}z}\Phi_2^*(n,\bft,z)\cdot \iota_{\La}\Delta^{-1}\Big(\Phi_2(n',\mathbf{t}',z)\cdot r_{i,n'}(\mathbf{t}')\Big).
        \end{aligned}
    \end{align*}
    and 
    \begin{align*}
        &q_{i,n}(\mathbf{t}+[z^{-1}]_1)\Phi_1^*(n-1,\mathbf{t},z)=-z\cdot\iota_{\La}\Delta^{-1}\Big(\Phi^*_1(n,\mathbf{t},z)\cdot q_{i,n}(\mathbf{t})\Big),\\
        &q_{i,n-1}(\mathbf{t}+[z]_2)\Phi_2^*(n-1,\mathbf{t},z)=\iota_{\La^{-1}}\Delta^{-1}\Big(\Phi_2^*(n,\mathbf{t},z)\cdot q_{i,n}(\mathbf{t})\Big),\\
        &r_{i,n-1}(\mathbf{t}-[z^{-1}]_1)\Phi_1(n,\mathbf{t},z)=z\cdot\iota_{\La^{-1}}\Delta^{-1}\Big(\Phi_1(n,\mathbf{t},z)\cdot r_{i,n}(\mathbf{t})\Big),\\
        &r_{i,n}(\mathbf{t}-[z]_2)\Phi_2(n,\mathbf{t},z)=-\iota_{\La}\Delta^{-1}\Big(\Phi_2(n,\mathbf{t},z)\cdot r_{i,n}(\mathbf{t})\Big).
    \end{align*}
    Similarly,
    \begin{align*}
        &\oint_{C_\infty}\frac{dz}{2\pi {\bf i}z}\iota_{\La^{-1}}\Delta^{-1}\Big(\Phi_1(n,\mathbf{t},z)\cdot \tilde{r}_{i,n}(\mathbf{t})\Big)\iota_{\La}\Delta^{-1}\Big(\Phi^*_1(n',\mathbf{t}',z)\cdot \tilde{q}_{i,n'}(\mathbf{t}')\Big)\\
        &-\oint_{C_0}\frac{dz}{2\pi {\bf i}z}\iota_{\La}\Delta^{-1}\Big(\Phi_2(n,\mathbf{t},z)\cdot \tilde{r}_{i,n}(\mathbf{t})\Big)\iota_{\La^{-1}}\Delta^{-1}\Big(\Phi_2^*(n',\mathbf{t}',z)\cdot \tilde{q}_{i,n'}(\mathbf{t}')\Big)\\
        =&-\Delta^{-1}\Big(\tilde{q}_{i,n}(\bft)\tilde{r}_{i,n}(\bft)\Big)+\Delta'^{-1}\Big(\tilde{q}_{i,n'}(\bft')\tilde{r}_{i,n'}(\bft')\Big),
    \end{align*}
    and 
    \begin{align*}
        &\Delta^{-1}\Big(\tilde{q}_{i,n}(\bft-[z^{-1}]_1)\tilde{r}_{i,n}(\bft-[z^{-1}]_1)\Big)-\Delta^{-1}\Big(\tilde{q}_{i,n}(\bft)\tilde{r}_{i,n}(\bft)\Big)=-z^{-1}\tilde{q}_{i,n}(\bft)\tilde{r}_{i,n-1}(\bft-[z^{-1}]_1),\\
        &\Delta^{-1}\Big(\tilde{q}_{i,n+1}(\bft-[z]_2)\tilde{r}_{i,n+1}(\bft-[z]_2)\Big)-\Delta^{-1}\Big(\tilde{q}_{i,n}(\bft)\tilde{r}_{i,n}(\bft)\Big)=\tilde{q}_{i,n}(\bft)\tilde{r}_{i,n}(\bft-[z]_2).
    \end{align*}
    Here $\Phi_a$ are Toda wave functions and $\Phi_a^*$ are adjoint wave functions, defined by 
    \begin{align*}
        &\Phi_1(n,\mathbf{t},z)=\frac{\tau_{n}(\mathbf{t}-[z^{-1}]_1)}{\tau_{n}(\mathbf{t})}z^ne^{\xi(t^{(1)},z)},\quad \Phi_2(n,\mathbf{t},z)=\frac{\tau_{n+1}(\mathbf{t}-[z]_2)}{\tau_{n}(\mathbf{t})}z^ne^{\xi(t^{(2)},z^{-1})},\label{eqvfunc2}\\
        &\Phi_1^*(n,\mathbf{t},z)=\frac{\tau_{n+1}(\mathbf{t}+[z^{-1}]_1)}{\tau_{n+1}(\mathbf{t})}z^{-n}e^{-\xi(t^{(1)},z)},\quad \Phi_2^*(n,\mathbf{t},z)=\frac{\tau_{n}(\mathbf{t}+[z]_2)}{\tau_{n+1}(\mathbf{t})}z^{-n}e^{-\xi(t^{(2)},z^{-1})}.
    \end{align*} 
    Note that the definition of the Toda adjoint wave functions $\Phi^*_a$ here is slightly different from that in \cite{LiuYue2024}. The remaining proof is similar to Theorem \ref{Thm:DDelta}.
\end{proof}

\section{Conclusions and Discussions}\label{sect:Conclusions}

In this paper, the GBMT hierarchy is constructed by imposing the following constraint on the mToda Lax operators $(L_1,L_2)$,
\[
    L_1(n)^M=L_2(n)^N+\sum_{l\in\Z}\sum_{i=1}^{m}q_{i,n}\La^lr_{i,n+1}\Delta,\qquad \big(L_1^M(n)+L_2^N(n)\big)(1)=0,
\]
which is equivalent to 
\begin{align}
        &L_1(n)^M=(L_1(n)^M)_{\Delta,\geq1}+(L_2(n)^N)_{\Delta^*,\geq1}+\sum_{i=1}^{m}q_{i,n}\cdot\iota_{\La^{-1}}\Delta^{-1}\cdot r_{i,n+1}\cdot\Delta,\label{eq:consL1L2'1}\\
        &L_2(n)^N=(L_1(n)^M)_{\Delta,\geq1}+(L_2(n)^N)_{\Delta^*,\geq1}+\sum_{i=1}^{m}q_{i,n}\cdot\iota_{\La}\Delta^{-1}\cdot r_{i,n+1}\cdot\Delta.\label{eq:consL1L2'2}
\end{align}
In terms of the mToda tau pair $(\tau_{0,n},\tau_{1,n})$, the GBMT constraint \eqref{eq:consL1L2'1}--\eqref{eq:consL1L2'2} means 
\[
    D_{M,N}\Delta\log\tau_{0,n}=-\sum_{i=1}^{m}q_{i,n}\Delta(r_{i,n}),\qquad D_{M,N}\Delta\log\tau_{1,n}=\sum_{i=1}^{m}\Delta(q_{i,n})r_{i,n+1}.
\]
The corresponding GBMT tau pair $(\tau_{0,n},\tau_{1,n})$ satisfy 
\begin{align}
    &\begin{aligned}\label{eq:bieqtau1}
        &\oint_{C_\infty}\frac{dz}{2\pi\textnormal{\bf i}z}z^{M+n-n'}\tau_{0,n}(\bft-[z^{-1}]_1)\tau_{1,n'}(\bft'+[z^{-1}]_1)e^{\xi(\bft^{(1)}-\bft'^{(1)},z)}\\
        &\quad +\oint_{C_0}\frac{dz}{2\pi\textnormal{\bf i}z}z^{N+n-n'+1}\tau_{0,n+1}(\bft-[z]_2)\tau_{1,n'-1}(\bft'+[z]_2)e^{\xi(\bft^{(2)}-\bft'^{(2)},z^{-1})}=\sum_{i=1}^{m}\rho_{i,n}(\bft)\sigma_{i,n'}(\bft'),
    \end{aligned}
\end{align}
In addition, it is found that the GBMT hierarchy is related to GBT hierarchy
\begin{equation}\label{eq:consGBT}
    \LL_1^M=\LL_2^{N}+\sum_{j\in \mathbb Z}\sum_{i=1}^{m}\tilde{q}_{i,n}\Lambda^j\tilde{r}_{i,n}
\end{equation}
via Miura links $\LL_aT_b=T_bL_a~(a,b=1,2)$ with $T_1=c_0^{-1}(n),~ T_2=c_0^{-1}(n+1)\Delta$. Moreover, the GBT constraint \eqref{eq:consGBT} is equivalent to 
\[
    D_{M,N}\Delta\big(\log\tau_n(\bft)\big)=\sum_{i=1}^{m}\tilde{q}_{i,n}(\bft)\tilde{r}_{i,n}(\bft),
\]
where $\tau_n(\bft)$ is the Toda tau function.

Here we would like to point out that the GBMT hierarchy is the generalization of the following systems.
\begin{itemize}
    \item The bigraded modified Toda hierarchy\cite{Yang2024},
    \[
        L_1^M=L_2^N,\qquad (L_1^M+L_2^N)(1)=0.
    \]
    \item The constrained modified KP hierarchy\cite{Cheng2018,WuCheng2022,Oevel1998},
    \[
        \fL^k=(\fL^k)_{\geq 1}+\sum_{j=1}^m \mathfrak{q}_j\p^{-1}\mathfrak{r}_j\p.
    \]
    Notice that $k=M+N$, if we let $\bft'^{(2)}=\bft^{(2)}, N=0, n'=n$ in \eqref{eq:bieqtau1}, then we get the bilinear equation \eqref{eq:consmKPBiEq1} of the cmKP hierarchy. Thus the GBMT hierarchy can be viewed as the two-component generalization of the cmKP hierarchy.
    \item The constrained modified discrete KP hierarchy \cite{Jian2018},
    \begin{equation}\label{eq:cmdKP}
        \bar{L}^k=(\bar{L}^k)_{\Delta,\geq1}+\sum_{i=1}^{m}q_i\cdot\iota_{\La^{-1}}\Delta^{-1}\cdot r_i\cdot \Delta.
    \end{equation}
    In fact if we let $N=0$ for the first equation of \eqref{eq:consL1L2'1}, we get 
    \[
        L_1(n)^M=(L_1(n)^M)_{\Delta,\geq1}+\sum_{i=1}^{m}q_{i,n}\cdot\iota_{\La^{-1}}\Delta^{-1}\cdot r_{i,n+1}\cdot\Delta,
    \]
    which is just the constrained modified discrete KP constraint \eqref{eq:cmdKP}.
\end{itemize}
Since the above integrable systems have many important applications in mathematical physics and integrable systems, we believe that the GBMT hierarchy should also be of great importance.
\\
\\

\noindent{\bf Acknowledgements}:

This work is supported by National Natural Science Foundation of China (Grant Nos. 12171472 and 12261072).\\

\noindent{\bf Conflict of Interest}:

 The authors have no conflicts to disclose.\\

\noindent{\bf Data availability}:

Date sharing is not applicable to this article as no new data were created or analyzed in this study.

\end{document}